\documentclass[11pt,a4paper,reqno]{amsart}
\usepackage{multirow}
\usepackage{cite}
\usepackage{amsmath}
\usepackage{amsfonts}
\usepackage{amssymb}
\usepackage{amsthm}
\usepackage{newlfont}
\usepackage{a4}
\usepackage{enumerate}
\usepackage{doi}

\usepackage{graphicx,color}
\usepackage{todonotes}

\theoremstyle{definition}
\newtheorem{defn}{Definition}[section]
\newtheorem{thm}[defn]{Theorem}

\newtheorem{tvr}[defn]{Proposition}

\theoremstyle{remark}
\newtheorem{example}{Example}[section]

\usepackage{bbm}

\allowdisplaybreaks

\newlength{\defbaselineskip}
\newcommand{\setlinespacing}[1]%
           {\setlength{\baselineskip}{#1 \defbaselineskip}}

\newcommand{\hart}{\zeta}
\newcommand{\si}{\sigma}
\renewcommand{\i}{\mathrm{i}}
\newcommand{\map}{\rightarrow}

\newcommand{\id}{\mathfrak{1}}

\newcommand{\q}{\quad}

\renewcommand{\epsilon}{\varepsilon}

\newcommand{\ep}{\varepsilon}
\newcommand{\la}{\lambda}
\newcommand{\al}{\alpha}
\newcommand{\om}{\omega}
\renewcommand{\rho}{\varrho}
\renewcommand{\phi}{\varphi}

\newcommand{\R}{{\mathbb{R}}}

\newcommand{\N}{{\mathbb N}}

\newcommand{\Z}{\mathbb{Z}}
\newcommand{\C}{\mathbb{C}}

\newcommand{\seq}[1]{\left<#1\right>}
\newcommand{\set}[2]{\left\{#1 \, |\, #2 \right\}}
\newcommand{\setb}[2]{\left\{#1 \, \mid\, #2 \right\}}
\newcommand{\setm}[2]{\left\{#1 \,\, \big|\,\, #2 \right\}}
\newcommand{\abs}[1]{\left\vert#1\right\vert}
\newcommand{\wt}{\widetilde}

\newcommand{\sca}[2]{\langle #1,\, #2\rangle}

\begin{document}

\title[Dual-root lattice discretization]
{Dual-root lattice discretization of Weyl orbit functions}

\author[J. Hrivn\'{a}k]{Ji\v{r}\'{i} Hrivn\'{a}k$^{1}$}
\author[L. Motlochov\'{a}]{Lenka Motlochov\'{a}$^{1}$}

\date{\today}
\begin{abstract}\small
Four types of discrete transforms of Weyl orbit functions on the finite point sets are developed. The point sets are formed by intersections of the dual-root lattices with the fundamental domains of the affine Weyl groups. The finite sets of weights, labelling the orbit functions, obey symmetries of the dual extended affine Weyl groups. Fundamental domains of the dual extended affine Weyl groups are detailed in full generality. Identical cardinality of the point and weight sets is proved and explicit counting formulas for these cardinalities are derived. Discrete orthogonality of complex-valued Weyl and real-valued Hartley orbit functions over the point sets is established and the corresponding discrete Fourier-Weyl and Hartley-Weyl transforms are formulated.
\end{abstract}

\maketitle
\noindent
$^1$ Department of Physics, Faculty of Nuclear Sciences and Physical Engineering, Czech Technical University in Prague, B\v{r}ehov\'a~7, CZ-115 19 Prague, Czech Republic
\vspace{10pt}

\noindent
E-mail: jiri.hrivnak@fjfi.cvut.cz, lenka.motlochova@fjfi.cvut.cz

\bigskip

\smallskip


\section{Introduction}

The purpose of this article is to extend the collection of discrete Fourier transforms of Weyl orbit functions on Weyl group invariant lattices \cite{HP,HMPdis,HW2}. A finite fragment of a refinement of the classical dual root lattice \cite{Bour} serves as the starting set of points over which the discrete orthogonality of four types of complex Weyl orbit functions \cite{KP1,KP2,MMP} is developed. The entire resulting transform formalism produces a real-valued multidimensional Weyl group invariant generalizations of the one-dimensional discrete Hartley transform~\cite{Brac}. 

The antisymmetric and symmetric exponential orbit sums over Weyl groups form a standard part of the theory of Lie algebras and their representations \cite{Bour}. From the viewpoint of the Coxeter groups theory, Weyl groups cover all finite crystallographic reflection groups \cite{H2}.  Depending on the type of the underlying crystallographic root system, two or four sign homomorphisms exist \cite{MMP}. Each sign homomorphism determines signs in the exponential sums and thus generates for each Weyl group two or four types of complex special functions.  
Lattice shift and Weyl group invariance of the resulting Weyl orbit functions generalize periodicity and boundary behaviour of the standard cosine and sine functions of one variable. Investigating Weyl orbit functions as special functions, the results range from generalizations of continuous multivariate Fourier transforms in \cite{KP1,KP2} to generalized Chebyshev polynomial methods \cite{xuAd,MPcub}. Discrete Fourier methods are comprehensively studied for Weyl orbit functions \cite{HP,HMPdis,HW2} as well as for their multivariate Chebyshev polynomial generalizations \cite{diejen,HMPcub,MMP,MPcub}. The refinement of the dual weight lattice intersected with the fundamental domain of the affine Weyl group form a finite point set on which the majority of the discrete Fourier and Chebyshev methods is developed \cite{HP,HMPdis,xuAd,MMP}. This choice of the point set generates symmetries of labels of orbit functions governed by the dual affine Weyl group  \cite{HP}. Choosing as the starting point set the refinement of the weight lattice produced different argument and label symmetries, both controlled by the same affine Weyl group \cite{HW2}. 

The dual root and root lattices constitute the last classical Weyl group invariant lattices for which the inherent Fourier methods have not yet been studied. Several apparent relative difficulties, surmounted in the present paper for the dual root lattice, stem from the fact that the label symmetries are in this case determined by the dual extended affine Weyl group. 
Firstly, even though the structure of the extended affine Weyl group and its dual version is detailed already in \cite{Bour}, relevant rigorous results about their fundamental domains appeared much later in \cite{Komr}. Moreover, these fundamental domains, essential for selecting the labels of orbit functions in discrete transforms, are determined in \cite{Komr} only up to their boundaries. A~uniform description of the fundamental domains from \cite{Komr}, including a unique layout of the boundary points, is achieved in the present paper by introducing lexicographical ordering on the Kac coordinates \cite{Kac}. Secondly, the main challenge poses linking the number of weights, found in the fundamental domain of the dual extended affine Weyl group, with the number of points from the refined dual root lattice, lying in the fundamental domain of the affine Weyl group.  Both sets form topologically distant finite subsets of the underlying Euclidean space and their common cardinality constitutes a novel invariant characteristic of the crystallographic root systems and corresponding simple Lie algebras. Determining the cardinality of these sets in full generality requires invoking and extending concepts from the theory of invariant polynomials \cite{invtheory,H2}. Common cardinality of the point and weight sets guarantees in turn the existence of both complex-valued Fourier-Weyl and real-valued Hartley-Weyl discrete transforms.   

Both complex and real types of developed Fourier-like transforms significantly enhance the collection of available Weyl group invariant discrete transforms \cite{HP,HMPdis,xuAd,MMP}. The topology of the point sets of the dual root lattice discretization and their relative position with respect to the fundamental domain of the affine Weyl group substantially diverge from the locations of the original dual weight lattice points. Moreover, fundamentally novel options are generated by combining both dual root and dual weight discretizations to produce discrete transforms on generalized and composed grids. These options encompass transforms on the refined dual weight lattice with points from the dual root lattice omitted, such as two-variable transforms on the  honeycomb lattice. Assuming Weyl symmetric or antisymmetric boundary conditions in mechanical graphene eigenvibrations model \cite{CsTi} or in quantum field lattice models \cite{DrSa} potentially yields novel Weyl invariant solutions and dispersion relations in terms of four types of extended Weyl orbit functions.    
 Real-valued multivariate Hartley-Weyl versions of the transforms augment the application potential of discrete Hartley transforms in pattern recognition \cite{Bhar}, geophysics \cite{Kuhl}, signal processing \cite{Paras, Pusch}, optics \cite{Liu} and measurement \cite{Sun}. The unitary matrices of the derived transforms also permit the construction of novel analogues of the Kac-Peterson matrices from the conformal field theory \cite{HW2}.     

The paper is organized as follows. In Section 2, the necessary facts concerning root systems and invariant lattices of Weyl groups are recalled. Section 3 contains a description of infinite extensions of Weyl groups. In Section 4, fundamental domains of infinite extensions of Weyl groups and the corresponding invariant polynomials are detailed. Section 5 is devoted to the study of finite sets of points and weights. The identical cardinality of these sets is proven and explicit counting formulas for the cardinalities are listed in full generality. Section~6 describes Weyl and Hartley orbit functions together with their discrete orthogonality and discrete Fourier transforms. Comments and follow-up questions are contained in the last section.

\section{Invariant lattices of Weyl groups}

\subsection{Root systems and Weyl groups}\

The notation used in this article is based on papers \cite{HP,HMPdis}. The purpose of this section is to extend this notation and recall other pertinent details \cite{Bour}.
Each simple Lie algebra from the classical four series $A_n\,(n\geq1)$, $B_n\,(n\geq3)$, $C_n\,(n\geq2)$, $D_n\,(n\geq4)$ and from the five exceptional cases $E_6,E_7,E_8,F_4, G_2$ determines its set $\Delta=\{\al_1,\dots,\al_n\}$ of simple roots \cite{Bour,VO}. For the cases of simple Lie algebras with two different root-lengths, the set $\Delta$ is disjointly decomposed into $\Delta_s$ of short simple roots and $\Delta_l$ of long simple roots, 
$$\Delta=\Delta_s\cup\Delta_l.$$ 
The set $\Delta$ forms a non-orthogonal basis of the Euclidean space $\R^n$ with the standard scalar product $\seq{\cdot,\cdot}$. 
To every simple root $\al_i\in \Delta $ corresponds a reflection $r_i$ given by the formula
\begin{equation*}
r_i a\equiv a-\frac{2\seq{a,\al_i}}{\seq{\al_i,\al_i}}\al_i,\q a\in\R^n\,.
\end{equation*}
Reflections $r_i, \, i\in \{1,\dots,n\}$ generate an irreducible Weyl group 
$$W\equiv\seq{r_1,\dots,r_n}$$
and $W$ in turn generates the entire root system $\Pi\equiv W\Delta$. The set of simple roots $\Delta$ induces a partial ordering $\leq$ on $\R^n$ such that for $\la,\nu \in \R^n$ it holds $\nu\leq \la$ if and only if $\la-\nu = k_1\al_1+\dots+ k_n \al_n$ with $k_i \in \Z^{\geq 0}$ for all $i\in \{1,\dots,n\}$. There exists a unique root $\xi \in \Pi$, highest with respect to this ordering, of the form $$\xi=m_1\al_1+\dots+m_n\al_n.$$ To each simple root $\al_i\in\Delta$ relates the dual simple root given by
$$\alpha_i^\vee=\frac{2\al_i}{\seq{\al_i,\al_i}}.$$
The set of dual simple roots $\Delta^\vee = \{\al_1^\vee,\dots,\al_n^\vee\}$ generates the entire dual root system  $\Pi^\vee\equiv W\Delta^\vee$. The dual root system $\Pi^\vee$ contains the highest dual root with respect to the ordering induced by $\Delta^\vee$ of the form $$\eta=m^\vee_1\al^\vee_1+\dots+m^\vee_n\al^\vee_n.$$ The expansion coefficients of the highest root $m_1,\dots,m_n $ and of the highest dual root $m ^\vee_1,\dots,m^\vee_n $, named the marks and the dual marks, respectively, are listed in Table~1 in \cite{HP}. Setting additionally $m_0=m_0^\vee=1$, the Coxeter number $m$ is given by
\begin{equation}\label{coxnum}
m=\sum_{i=0}^n m_i=\sum_{i=0}^n m^\vee_i . 
\end{equation}
The dual marks with unitary values determine an important index set $J \subset \{1,\dots,n\}$,
\begin{equation*}
J\equiv\set{1\leq i\leq n}{m^\vee_i=1}.
\end{equation*}    
Note that $J$ is empty for the simple Lie algebras $E_8,F_4$ and $G_2$.

Minimal number generators $r_i$, necessary to generate an element $w\in W,$ is called the length $l(w)$ of $w$. There is a unique longest element $w_0 \in W$, also called the opposite involution, and for its length it holds that 
\begin{equation*}
l(w_0)= \frac{|\Pi|}{2}.	
\end{equation*}
For root systems with even Coxeter numbers are the opposite involutions of the form
$$w_0 = (r_{1}\dots r_{n})^\frac{m}{2}.$$
An important parabolic subgroup $W_i$ of the Weyl group $W$ is obtained by omitting  $r_i$ from the set of generators of $W,$
$$W_i\equiv\seq{r_1,\dots,r_{i-1},r_{i+1},\dots,r_n}.$$ 
The subgroup $W_i$ forms the Weyl group corresponding to the root system $\Pi_i\equiv W_i(\Delta\setminus \al_i)$. 
The opposite involution in $W_i$ with respect to the root system $\Pi_i$ is denoted by $w_i$.

\subsection{Invariant lattices}\
 
Four classical Weyl group invariant lattices \cite{Bour} comprise the root lattice, the dual weight lattice, the dual root lattice and the weight lattice. The root lattice $Q$ is the $\Z$-span of the set of simple roots $\Delta$,
 $$Q=\Z\al_1+\dots +\Z\al_n.$$
The dual weight lattice is $\Z-$dual to the root lattice $Q$, 
$$P^\vee=\set{\omega^\vee\in\R^n}{\seq{\omega^\vee,\al_i}\in\Z,\,\forall \al_i\in\Delta}=\Z\omega_1^\vee+\dots+\Z\omega_n^\vee,$$ 
with the dual fundamental weights $\omega_i^\vee$ given by $$\seq{\omega^\vee_i,\al_j}=\delta_{ij}.$$
 The dual root lattice is the $\Z$-span of the set of dual simple roots $\Delta^\vee,$ 
 $$Q^\vee=\Z\al_1^\vee+\dots +\Z\al_n^\vee.$$
The weight lattice is $\Z-$dual to the dual root lattice $Q^\vee$, 
\begin{equation}\label{P}
P=\set{\omega\in\R^n}{\seq{\omega,\al^\vee_i}\in\Z,\,\forall \al^\vee_i\in\Delta^\vee}=\Z\omega_1+\dots+\Z\omega_n,	
\end{equation}
with the fundamental weights $\omega_i$ given by $$\langle\omega_i,\al^\vee_j\rangle=\delta_{ij}.$$
The weight lattice $P$ is partitioned into $|J|+1$ components and this decomposition consists of the root lattice $Q$ and its $|J|$ shifted copies, 
\begin{equation*}
P=Q\cup  \bigcup_{i\in J}(\omega_i+Q).
\end{equation*}
The Cartan matrix $C$ with entries given by
\begin{equation}\label{Cart}
C_{ij}=\seq{\al_i,\al^\vee_j} 	
\end{equation}
relates the simple roots and fundamental weights as well as their dual versions via formulas
\begin{equation}\label{alom}
\omega_i=\sum_{j=1}^n C^{-1}_{ij}\al_j,\q \omega^\vee_i=\sum_{j=1}^n \left(C^T_{ij}\right)^{-1}\al^\vee_j.
\end{equation}
The determinant $c$ of the Cartan matrix $C$ determines the index of connection of the root system $\Pi$ and the order of the quotient group $P/Q$,
$$c=\det C= \abs{P/Q}=\abs{J}+1. $$ 

\subsection{Sign homomorphisms}\

Recall from \cite{HMPdis} that a homomorphism $\sigma :W\mapsto \{\pm 1\} $ is called a sign homomorphism.
The identity $\id$ and the determinant $\si^e$ sign homomorphisms, which exist for any Weyl group $W,$ are given on the generating reflections $r_i$ as
\begin{alignat*}{2}
& \id (r_i) &=& 1,\\ 
& \sigma^e (r_i) &=& -1.
\end{alignat*} 
For the root systems with two lengths of roots, the short and long sign homomorphisms $\sigma^s $ and $\sigma^l$ are defined as
\begin{alignat*}{1}
 \sigma^s (r_i) &= \left\{ 
  \begin{array}{l l}
    -1 & \quad \text{if $\alpha_i \in \Delta_s, $}\\
     1 & \quad \text{otherwise,}
  \end{array} \right. \\
	\sigma^l (r_i) &= \left\{ 
  \begin{array}{l l}
    -1 & \quad \text{if $\alpha_i \in \Delta_l, $}\\
     1 & \quad \text{otherwise.}
  \end{array} \right. 
\end{alignat*}

To each sign homomorphism $\sigma$ is attached a vector $\rho^\sigma=\rho^\sigma_1\omega_1+\dots+\rho^\sigma_n\omega_n\in P$ defined by   
\begin{equation}\label{rhoi}
\begin{alignedat}{4}
&\rho_i^\id\equiv0&\q&i=1,\dots,n,\\
&\rho_i^{\sigma^e}\equiv1&\q&i=1,\dots,n,\\
&\rho_i^{\sigma^s}\equiv1&\q&\alpha_i\in\Delta_s,&\q\q& \rho_i^{\sigma^s}\equiv0&\q& \alpha_i\in\Delta_l,\\
&\rho_i^{\sigma^l}\equiv0&\q&\alpha_i\in\Delta_s,&\q\q& \rho_i^{\sigma^l}\equiv1&\q&\alpha_i\in\Delta_l.
\end{alignedat}
\end{equation}
The vector $\rho^{\sigma^e}$ becomes the standard $\rho$ vector defined as half-sum of the positive roots. Zero coordinates $\rho^\sigma_0$ of the vectors $\rho^\sigma$ are for convenience defined  by
$$\rho^\id_0\equiv0,\q \rho^{\sigma^e}_0\equiv1,\q\rho^{\sigma^s}_0\equiv1,\q\rho^{\sigma^l}_0\equiv0.$$
Furthermore, to each sign homomorphism $\sigma$ is associated a generalized Coxeter number $m^\sigma$ by defining relation
\begin{equation}\label{msigma}
m^\sigma= \sum_{i=0}^n m^\vee_i \rho^\sigma_i.
\end{equation}
Note that $m^\id=0$ and the number $m^{\sigma^e}=m$ coincides with the standard Coxeter number. The short $m^{\sigma^s}$ and long $m^{\sigma^l}$ Coxeter numbers are tabulated in \cite{HMPdis}. 

\section{Extensions of Weyl groups}

\subsection{Affine Weyl groups}\

The affine Weyl group is a semidirect product of the group of translations $Q^\vee$ and $W,$ 
\begin{equation}\label{Waff}
W^{\mathrm{aff}}= Q^\vee\rtimes W.
\end{equation}
For $q^\vee \in Q^\vee$ and $w\in W$, any element $T(q^\vee)w \in W^{\mathrm{aff}}$ acts on $\R^n$ as $$T(q^\vee)w\cdot a\equiv wa+q^\vee, \q a\in \R^n.$$ 
The standard retraction homomorphism $\psi :W^{\mathrm{aff}}\map W$ of the semidirect product \eqref{Waff}  is given by 
\begin{equation*}
\psi(T(q^\vee)w)=w.
\end{equation*}

The fundamental domain $F\subset\R^n$ of $W^{\mathrm{aff}}$ 
is a simplex explicitly given by
$$F=\set{a_1\omega_1^\vee+\dots+a_n\omega_n^\vee}{a_0+m_1a_1+\dots+m_na_n=1,a_i\geq0,i=0,\dots,n}.$$
The stabilizer $\mathrm{Stab}_{W^\mathrm{aff}}(a)$ is a subgroup of $W^\mathrm{aff}$ stabilizing $a\in \R^n,$
\begin{equation*}
\mathrm{Stab}_{W^{\mathrm{aff}}}(a) = \setb{w^{\mathrm{aff}}\in W^{\mathrm{aff}}}{w^{\mathrm{aff}}a=a},
\end{equation*}
and the function $\varepsilon:\R^n\map \N$ is defined by 
\begin{equation}\label{ep}
\varepsilon(a)=\frac{\abs{W}}{\abs{\mathrm{Stab}_{W^{\mathrm{aff}}}(a)}}.
\end{equation} 
The stabilizers $\mathrm{Stab}_{W^{\mathrm{aff}}}(a) $ and $\mathrm{Stab}_{W^{\mathrm{aff}}}(w^{\mathrm{aff}}a) $ are conjugated
and therefore the function $\varepsilon$ is $W^{\mathrm{aff}}$ invariant
\begin{equation}\label{epinv}
\varepsilon(a)=\varepsilon(w^{\mathrm{aff}}a), \q w^{\mathrm{aff}} \in W^{\mathrm{aff}}.
\end{equation} 

The standard action of $W$ on the torus $\R^n/Q^{\vee}$ generates for $x\in \R^n/Q^{\vee}$ its isotropy groups $\mathrm{Stab} (x)$
and orbits $W x$ of orders
\begin{equation*}
\wt \ep(x)\equiv |Wx|, \q x\in \R^n/Q^{\vee}.
\end{equation*}
The following three properties from Proposition 2.2 in \cite{HP} of the action of $W$ on the torus $\R^n/Q^{\vee}$ are essential,
\begin{enumerate}
\item For any $x\in \R^n/Q^{\vee}$, there exists $x'\in F \cap \R^n/Q^{\vee} $ and $w\in W$ such that
\begin{equation}\label{rfun1}
 x=wx'.
\end{equation}
\item If $x,x'\in F \cap \R^n/Q^{\vee} $ and $x'=wx$, $w\in W$, then
\begin{equation}\label{rfun2}
 x'=x=wx.
\end{equation}
\item If $x\in F \cap \R^n/Q^{\vee} $, i.e. $x=a+Q^{\vee}$, $a\in F$, then $\psi (\mathrm{Stab}_{W^{\mathrm{aff}}}(a))=\mathrm{Stab}(x)$ and
\begin{equation}\label{rfunstab}
\mathrm{Stab} (x) \cong \mathrm{Stab}_{W^{\mathrm{aff}}}(a).
\end{equation}
\end{enumerate}
Relation \eqref{rfunstab} grants that for $x=a+Q^{\vee}$, $a\in F$ it holds that
\begin{equation}\label{ept}
	\ep(a)= \wt\ep(x).
\end{equation}
Note that instead of $\wt\ep(x)$, the symbol $\ep(x)$ is used for $|Wx|$, $x\in F\cap\R^n/Q^{\vee} $ in \cite{HP,HMPdis}. The algorithm for calculation of the coefficients $\ep(x)$ is described in
\cite[\S 3.7]{HP}.

To each sign homomorphism $\si$ is associated a subset $F^\si \subset F$ of the form   
\begin{equation}\label{domainF2}
F^\sigma\equiv\setb{a\in F}{\sigma\circ\psi\left(\mathrm{Stab}_{W^\mathrm{aff}}(a)\right)=\{1\}}.
\end{equation}  
The sets $F^\sigma$ are detailed in \cite{HP,HMPdis} and described via   
non-negative symbols $a_i^\sigma $ defined by
\begin{equation*}
\begin{alignedat}{4}
&a^\mathfrak{1}_0,a^\mathfrak{1}_i\geq 0 &\q& \text{ if } \alpha_i\in\Delta,\\
&a^{\sigma^e}_0,a^{\sigma^e}_i> 0 &\q& \text{ if } \alpha_i\in\Delta,\\
&a^{\sigma^s}_0,a^{\sigma^s}_i\geq 0 &\q& \text{ if } \alpha_i\in\Delta_l, &\q\q & a^{\sigma^s}_i>0 &\q& \text{ if }\al_i\in\Delta_s,\\
&a^{\sigma^l}_0,a^{\sigma^l}_i> 0 &\q&\text{ if } \alpha_i\in\Delta_l, &\q\q & a^{\sigma^l}_i\geq 0 &\q&\text{ if }\al_i\in\Delta_s.
\end{alignedat}
\end{equation*} 
The explicit formulas for $F^\sigma$ are thus of the form
\begin{equation}\label{domainF}
F^\sigma=\set{a^\sigma_1\omega_1^\vee+\dots+a^\sigma_n\omega_n^\vee}{a^\sigma_0+m_1a^\sigma_1+\dots+m_na^\sigma_n=1}.
\end{equation}
Note that $F^\mathfrak{1}= F$ and $F^{\sigma^e}=\mathrm{int}(F)$.

\subsection{Dual affine Weyl groups}\

The dual affine Weyl group $W^{\mathrm{aff}}_Q$ is a semidirect product of the group of translations in the root lattice $Q$ and $W,$
\begin{equation}\label{WaffQ}
	W^{\mathrm{aff}}_Q\equiv Q\rtimes W.
\end{equation}
For $q\in Q$ and $w\in W$, any element $T(q)w \in W^{\mathrm{aff}}_Q$ acts on $\R^n$ as 
\begin{equation}\label{TQ}
	T(q)w\cdot b\equiv wb+q, \q b\in \R^n.
\end{equation}
The fundamental domain $F_Q\subset\R^n$ of the dual affine Weyl group $W^{\mathrm{aff}}_Q$, denoted by $F^\vee$ in \cite{HP,HMPdis}, is explicitly given by  
\begin{equation}\label{FQ}
F_Q\equiv\set{b_1\omega_1+\dots+b_n\omega_n}{b_0+m^\vee_1 b_1+\dots+m_n^\vee b_n=1,b_i\geq0,i=0,\dots,n}.
\end{equation}  
By identifying each $b=b_1\omega_1+\dots+b_n\omega_n\in F_Q$ with its Kac coordinates $[b_0,b_1,\dots,b_n]$ from \eqref{FQ}, 
$$b\equiv[b_0,b_1,\dots,b_n],$$
the lexicographic ordering on $F_Q$ is introduced in the following way. An element $b=[b_0,b_1,\dots,b_n]$ is lexicographically higher than $b'=[b'_0,b_1',\dots, b_n']$,
\begin{equation}\label{lex}
b>_\mathrm{lex} b',	
\end{equation}
if and only if $b_i>b_i'$ for the first $i=0,1,\dots,n$ where $b_i$ differs from $b_i'$.

The standard dual retraction homomorphism $\widehat\psi :W^{\mathrm{aff}}_Q\map W$ of the semidirect product \eqref{WaffQ}  is given by 
\begin{equation}\label{retQ}
\widehat\psi(T(q)w)=w.
\end{equation}
Similar to \eqref{domainF2}, four subsets of $F_Q^\si \subset F_Q$ are introduced by
$$F_Q^\sigma\equiv\setm{b\in F_Q}{\sigma\circ\widehat\psi\left(\mathrm{Stab}_{W_Q^\mathrm{aff}}(b)\right)=\{1\}}.$$
The domains $F_Q^\sigma$ are explicitly described  by 
\begin{equation}\label{FQS}
F_Q^\sigma=\set{b_1^\sigma\omega_1+\dots+b_n^\sigma\omega_n}{b_0^\sigma+m_1^\vee b_1^\sigma+\dots+m_n^\vee b^\sigma_n=1}
\end{equation}
with the symbols $b_0^\sigma,\dots,b_n^\sigma$ satisfying
\begin{equation}\label{fun1}
\begin{alignedat}{4}
&b^\mathfrak{1}_0,b^\mathfrak{1}_i\geq 0 &\q& \text{ if } \alpha_i\in\Delta,\\ &b^{\sigma^e}_0,b^{\sigma^e}_i> 0 &\q& \text{ if } \alpha_i\in\Delta,\\
&b^{\sigma^s}_i\geq 0 &\q& \text{ if } \alpha_i\in\Delta_l, &\q\q& b^{\sigma^s}_0,b^{\sigma^s}_i>0 &\q& \text{ if }\al_i\in\Delta_s,\\
&b^{\sigma^l}_i> 0 &\q&\text{ if } \alpha_i\in\Delta_l, &\q\q& b^{\sigma^l}_0,b^{\sigma^l}_i\geq 0 &\q&\text{ if }\al_i\in\Delta_s.
\end{alignedat}
\end{equation} 

\subsection{Extended dual affine Weyl groups}\

The extended dual affine Weyl group $W^{\mathrm{aff}}_P$ is a semidirect product of the group of translations $P$ and the Weyl group $W,$ 
\begin{equation}\label{WaffP}
	W^{\mathrm{aff}}_P\equiv P\rtimes W.
\end{equation}
For $p\in P$ and $w\in W$, extending the action \eqref{TQ} on $\R^n$  to elements $T(p)w \in W^{\mathrm{aff}}_P$  yields $$T(p)w\cdot b\equiv wb+p, \q b\in \R^n.$$ 
The extended dual retraction homomorphism $\widehat\psi :W^{\mathrm{aff}}_P\map W$ of the semidirect product \eqref{WaffP}  is a natural extension of \eqref{retQ} given by 
\begin{equation*}
\widehat\psi(T(p)w)=w.
\end{equation*}

Introducing the subgroup $\Gamma \subset W_P^\mathrm{aff}$ of all elements of $W_P^\mathrm{aff}$ leaving the fundamental domain $F_Q$ of $W^{\mathrm{aff}}_Q$ invariant
\begin{equation}\label{Gamma0}
\Gamma=\set{\gamma\in W^{\mathrm{aff}}_P}{\gamma\cdot F_Q=F_Q}	
\end{equation}
allows to represent $W^{\mathrm{aff}}_P$ as a semidirect product 
\begin{equation}\label{extgroup}
W^{\mathrm{aff}}_P=W^{\mathrm{aff}}_Q\rtimes \Gamma.
\end{equation}
Explicit structure of the abelian group $\Gamma$ is detailed in \cite[Ch.VI,\S2]{Bour} as
\begin{equation}\label{gamma}
\Gamma=\setm{1,\gamma_i\in W^{\mathrm{aff}}_P}{i\in J},\q \gamma_i\equiv T(\omega_i)w_iw_0,\q\abs{\Gamma}=c.
\end{equation} 
The action of $\Gamma$ on $F_Q$ assigns to each fixed $\gamma\in\Gamma$ a bijection on $F_Q$ given on Kac coordinates $[b_0,\dots,b_n]$ by  
\begin{equation}\label{akcegamma}
\gamma\cdot b=\gamma\cdot [b_0,\dots,b_n]=[b_{\pi_\gamma(0)},\dots,b_{\pi_\gamma(n)}]
\end{equation}
where $\pi_\gamma$ denotes a permutation of the index set $\{0,1,\dots,n\}$. These permutations of the Kac coordinates $[b_0,\dots,b_n]$, which determine the group $\Gamma$, are specified for every simple Lie algebra in Table~\ref{gamtab}.
Direct analysis of the permutations $\gamma\in\Gamma$ in Table~\ref{gamtab} on the $\rho^\sigma$ vector \eqref{rhoi} yields for its coordinates $\rho^\sigma_i$ that  
\begin{equation}\label{gammarho}
\rho^\sigma_{\pi_\gamma(i)}=\rho^\sigma_i, \q i \in \{0,1,\dots,n \}.	
\end{equation}

{\small
\bgroup
\def\arraystretch{1.5}
\begin{table}
\begin{tabular}{|c||c|c|c|c|c|c|}
\hline
&\multirow{2}{*}{$\Gamma$}&\multirow{2}{*}{$\gamma$}&\multirow{2}{*}{$\gamma\cdot [b_0,\dots,b_n]$}&\multicolumn{3}{|c|}{$\sigma\circ\widehat{\psi}(\gamma)$}\\ \cline{5-7}
&&&&$\sigma^e$&$\sigma^s$&$\sigma^l$\\
\hline\hline
$A_n$&$\Z_{n+1}$&$\gamma_i$&$[b_{n-i+1},b_{n-i+2},\dots, b_n,b_0,b_1,\dots,b_{n-i}]$&$(-1)^{ni}$&$-$&$-$\\\hline
$B_n$&$\Z_2$&$\gamma_n$&$[b_n,b_{n-1},\dots, b_1,b_0] $&$(-1)^{\tfrac{n(n+1)}{2}} $&$(-1)^n$&$(-1)^{\tfrac{(n-1)n}{2}}$\\\hline
$C_n$&$\Z_2$&$\gamma_1$&$[b_1,b_0,b_2,\dots, b_n] $&$-1$&$1$&$-1$\\\hline
\multirow{3}{*}{$D_{2k}$}&\multirow{3}{*}{$\Z_2\times\Z_2$}&$\gamma_1$&$[b_1,b_0,b_2,\dots,b_{2k-2},b_{2k},b_{2k-1}] $&$1$&$-$&$-$\\\cline{3-7}
&&$\gamma_{2k-1}$& $[b_{2k-1},b_{2k},b_{2k-2},\dots,b_2,b_0,b_1]$&$(-1)^k$&$-$&$-$\\\cline{3-7}
&&$\gamma_{2k}$&$[b_{2k},b_{2k-1},\dots,b_1,b_0]$&$(-1)^k$&$-$&$-$\\\cline{3-7}
\hline
\multirow{3}{*}{$D_{2k+1}$}&\multirow{3}{*}{$\Z_4$}&$\gamma_1$&$ [b_1,b_0,b_2,\dots,b_{2k-1},b_{2k+1},b_{2k}]$&$1$&$-$&$-$\\\cline{3-7}
&&$\gamma_{2k}$&$[b_{2k+1},b_{2k},\dots,b_2,b_0,b_1]$&$(-1)^k$&$-$&$-$\\\cline{3-7}
&&$\gamma_{2k+1}$&$[b_{2k},b_{2k+1},b_{2k-1},\dots,b_1,b_0]$&$(-1)^k$&$-$&$-$\\\cline{3-7}
\hline
\multirow{2}{*}{$E_6$}&\multirow{2}{*}{$\Z_3$}&$\gamma_1$&$ [b_1,b_5,b_4,b_3,b_6,b_0,b_2]$&$1$&$-$&$-$\\\cline{3-7}
&&$\gamma_5$&$[b_5,b_0,b_6,b_3,b_2,b_1,b_4] $&$1$&$-$&$-$\\\hline
$E_7$&$\Z_2$&$\gamma_6$&$[b_6,b_5,b_4,b_3,b_2,b_1,b_0,b_7] $&$-1$&$-$&$-$\\\hline
\end{tabular}
\bigskip
\caption{Non-trivial  abelian groups $\Gamma$, isomorphic to the groups in the second column, are listed. The action of each $\gamma\in \Gamma$, different from the identity, on the Kac coordinates is  specified and the corresponding values of sign homomorphisms on $\widehat{\psi}(\gamma)$ are given.}\label{gamtab}
\end{table}
\egroup}

Considering a resolution factor $M\in\N$, an important class of subgroups $\Gamma_M$ of $W_P^{\mathrm{aff}}$ is given by
\begin{equation}\label{gammaM}
\Gamma_M=\set{1,\gamma_{M,i}\in W^{\mathrm{aff}}_P}{i\in J},\q \gamma_{M,i}\equiv T(M\omega_i)w_iw_0,\q\abs{\Gamma_M}=c.
\end{equation} 
For each $M$, the subgroup $\Gamma_M \subset W_P^{\mathrm{aff}}$ is isomorphic to $\Gamma$ by an isomorphism
\begin{equation}\label{izo}
\Gamma_M\ni  \gamma_M \mapsto \gamma\in\Gamma	
\end{equation}
defined by assigning $\gamma_i\in\Gamma$ to $\gamma_{M,i}\in\Gamma_M$.
 The action of $\gamma_M  \in \Gamma_M $ on $b\in\R^n$ is directly related to action of $\gamma \in \Gamma$ by
\begin{equation}\label{acgamM}
\gamma_M\cdot b=M\;\left(\gamma\cdot\frac{b}{M}\right).
\end{equation}
Consequently, $\Gamma_M$ acts naturally on the magnified fundamental domain $MF_Q$. Setting for $b=b_1\om_1+\dots+b_n\om_n\in MF_Q$ the magnified Kac coordinates $b_0\equiv M-\left(m_1^\vee b_1+\dots+m_n^\vee b_n\right)$, this action is described by 
\begin{equation}\label{akcegamM}
\gamma_M\cdot b=\gamma_M\cdot[b_0,\dots,b_n]=[b_{\pi_\gamma(0)},\dots,b_{\pi_\gamma(n)}].
\end{equation} 
 To each element $\gamma_M\in\Gamma_M$ is thus assigned the same permutation $\pi_\gamma$  of $\{0,1,\dots,n\}$ from Table \ref{gamtab} as to the corresponding $\gamma\in\Gamma$ in \eqref{akcegamma}. 

\section{Fundamental domains and invariant polynomials}

\subsection{Stabilizers}\

The stabilizer $\mathrm{Stab}_{W_P^\mathrm{aff}}(b)$ is a subgroup of $W_P^\mathrm{aff}$ stabilizing $b\in \R^n$
and the related counting discrete function $h_{P,M}:\R^n\map \N$ is for any $M\in \N$ defined by 
\begin{equation}\label{hPM}
h_{P,M}(b)=\abs{\mathrm{Stab}_{W_P^{\mathrm{aff}}}\left(\frac{b}{M}\right)}.
\end{equation} 
Characterizing the structure of $\mathrm{Stab}_{W_P^\mathrm{aff}}(b)$ in the following proposition subsequently allows the calculation of the function $h_{P,M}$.
\begin{tvr}
For any $b\in F_Q$ it holds that
\begin{equation}\label{stabil}
\mathrm{Stab}_{W^\mathrm{aff}_P}(b)=\mathrm{Stab}_{W^\mathrm{aff}_Q} (b)\rtimes\mathrm{Stab}_{\Gamma}(b).
\end{equation}
\end{tvr}
\begin{proof}
The semidirect decomposition \eqref{extgroup}, where $W^{\mathrm{aff}}_Q$ is normal in $W^\mathrm{aff}_P$, directly guarantees that  $\mathrm{Stab}_{W^\mathrm{aff}_Q} (b)$ is a normal subgroup of $\mathrm{Stab}_{W^\mathrm{aff}_P}(b)$. Moreover, for any $b\in F_Q$ stabilized by both $w^{\mathrm{aff}} \in W^{\mathrm{aff}}_Q$ and $\gamma \in \Gamma$ is $w^{\mathrm{aff}} \gamma \in \mathrm{Stab}_{W^\mathrm{aff}_P}(b).$ 
Conversely, for any $w_P\in \mathrm{Stab}_{W^\mathrm{aff}_P} (b)$ exists from \eqref{extgroup} a unique decomposition $w_P=w_Q\gamma$ with $w_Q\in W^\mathrm{aff}_Q$ and $\gamma\in\Gamma$. 
Invariance \eqref{Gamma0} of $F_Q$ under the action of $\Gamma$ implies that $\gamma\cdot b \in F_Q$. Since the fundamental domain $F_Q$ contains only one point from each $W^\mathrm{aff}_Q-$orbit, the relation $$w_P\cdot b=w_Q\cdot(\gamma\cdot b)=b$$ forces $\gamma\cdot b=b$. 
\end{proof}

The isomorphism \eqref{izo} and relation \eqref{acgamM} imply for the stabilizers that 
\begin{equation}\label{Stgm}
\mathrm{Stab}_{\Gamma_M}(b)\simeq \mathrm{Stab}_\Gamma\left(\frac{b}{M}\right).
\end{equation}
The resulting counting formula for $h_{P,M}$ is deduced for $b\in MF_Q$ from relations \eqref{stabil} and \eqref{Stgm} as
\begin{equation*}
h_{P,M}(b)=\abs{\mathrm{Stab}_{W_Q^{\mathrm{aff}}}\left(\frac{b}{M}\right)}\cdot\abs{\mathrm{Stab}_{\Gamma_M}(b)}.	
\end{equation*}
The calculation procedure for $|\mathrm{Stab}_{W_Q^{\mathrm{aff}}}\left(b/M\right)|$ is described in \cite[\S 3.7]{HP}. The orders of stabilizers  $\mathrm{Stab}_{\Gamma_M}(b)$ are directly derived from the explicit form of permutations given by \eqref{akcegamma} in Table \ref{gamtab}.

Note that since $\gamma_M \in \Gamma_M$ and $\gamma\in \Gamma$ from \eqref{gamma} and \eqref{gammaM} differ only in the translation part, their retractions coincide
\begin{equation}\label{hatpsi}
\widehat\psi\left(\gamma_M\right)=\widehat\psi(\gamma).
\end{equation}
Consequently, the retraction of the stabilizers are also identical
\begin{equation}\label{stabgamM}
\widehat\psi\left(\mathrm{Stab}_{\Gamma_M}(b)\right)=\widehat\psi\left( \mathrm{Stab}_\Gamma\left(\frac{b}{M}\right)\right).
\end{equation}

\subsection{Fundamental domains}\

Significant results concerning the structure of the fundamental domain $F_P$ of the group $W^{\mathrm{aff}}_P$ are achieved in \cite{Komr}.  
Firstly, from the semidirect decomposition \eqref{extgroup} follows  that the fundamental domain $F_P$ coincides with a fundamental domain of the group $\Gamma$ acting on $F_Q$. Secondly, 
the interior of $F_P$ is determined as \cite{Komr}
$$\mathrm{int}(F_P)=\set{b\in \mathrm{int}(F_Q)}{\langle\eta+\alpha_i^\vee,b\rangle<1,\, i\in J}.$$
It is also asserted that the extended interior $\mathrm{ent}(F_P)$ of $F_P$ defined by
\begin{equation}\label{fptilde}
\mathrm{ent}(F_P)\equiv\set{b\in F_Q}{\langle\eta+\alpha_i^\vee,b\rangle<1,\, i\in J}
\end{equation}
forms a subset of $F_P$,  $\mathrm{ent}(F_P) \subset F_P$. 

In order to uniquely determine the remaining boundary points from $F_P\setminus \mathrm{ent}(F_P)$, note that the defining relation in \eqref{fptilde}
\begin{equation*}
b_i+b_1m^\vee_1+\dots+m^\vee_nb_n=\langle\eta+\alpha_i^\vee,b\rangle<1=b_0+m^\vee_1b_1+\dots+m^\vee_n b_n	
\end{equation*}
implies that $b_0>b_i,i\in J$. Taking into account explicit forms of $\Gamma-$permutations of $[b_0,\dots,b_n]$ from Table~\ref{gamtab} and the lexicographic ordering \eqref{lex}, the inequality $b_0>b_i,i\in J$ grants that the point $[b_0,\dots,b_n]$ is the lexicographically highest among the points lying in its $\Gamma-$orbit. Consequently, the fundamental domain $F_P$ is taken as such subset of $F_Q$ which contains the lexicographically highest point from each $\Gamma-$orbit of  $b\in F_Q$. This resulting form of the fundamental domain $F_P$, which contains exactly one point from each $W^{\mathrm{aff}}_P-$orbit of $\R^n,$ is summarized in the following theorem. 
\begin{thm}\label{thmlex} The set $F_P\subset \R^n$ defined by
\begin{equation}\label{lexorder}
F_P=\set{b\in F_Q}{b=\mathrm{max}_{>_{\mathrm{lex}} } \Gamma b}
\end{equation}
forms a fundamental domain of the extended dual affine Weyl group $W^{\mathrm{aff}}_P.$
\end{thm}

Four crucial subsets $F^\sigma_P$ of $F_P$ are for each sign homomorphism $\sigma$ defined by
\begin{equation}\label{FsiP}
F^\sigma_P\equiv\set{b\in F_P}{\sigma\circ\widehat\psi(\mathrm{Stab}_{W^{\mathrm{aff}}_P} (b))=\{1\}}.
\end{equation}
Recall also from \cite{Komr} that the stabilizer $\mathrm{Stab}_{W^{\mathrm{aff}}_P}(b)$ of any interior point $b \in \mathrm{int}(F_P)$ is trivial, therefore it holds that $\mathrm{int}(F_P)\subset F^\sigma_P$. The form of the sets  $F^\sigma_P$ is simplified in the following proposition.
\begin{tvr}
The sets $F^\sigma_P$, defined by \eqref{FsiP}, are of the form
\begin{equation}\label{fundP}
F_P^\sigma=\set{b\in F_P\cap F_Q^\sigma}{\sigma\circ\widehat\psi\left(\mathrm{Stab}_\Gamma(b)\right)=\{1\}}.
\end{equation}
\end{tvr}
\begin{proof}
Since both stabilizers $\mathrm{Stab}_{W^{\mathrm{aff}}_Q} (b)$ and $ \mathrm{Stab}_{\Gamma} (b) $ are subgroups of  $\mathrm{Stab}_{W^{\mathrm{aff}}_P} (b)$, it holds for any point $b\in F^\sigma_P$ that 
\begin{equation}\label{bstab}
\sigma\circ\widehat\psi(\mathrm{Stab}_{W^{\mathrm{aff}}_Q} (b))=\sigma\circ\widehat\psi\left(\mathrm{Stab}_\Gamma(b)\right)=\{1\},
\end{equation}
and thus $b\in F^\sigma_Q$. Conversely, for any $b \in F_P$ follows from \eqref{bstab} and the semidirect decomposition \eqref{stabil} that $\sigma\circ\widehat\psi(\mathrm{Stab}_{W^{\mathrm{aff}}_P} (b))=\{1\}$ and thus $b\in F^\sigma_P$.  
\end{proof}

Theorem \ref{thmlex} and relation \eqref{acgamM} guarantee that the magnified domain $MF_P$ is a fundamental domain of the action of $\Gamma_M$ on $MF_Q$.
 Relations \eqref{fundP} and \eqref{stabgamM} allow to express the magnified domains $MF_P^\sigma$,
\begin{equation}\label{MFP}
MF_P^\sigma=\setb{b\in M(F_P\cap F_Q^\sigma)}{\sigma\circ\widehat\psi\left(\mathrm{Stab}_{\Gamma_M}(b)\right)=\{1\}}.
\end{equation}
\begin{example}\label{exA3}
For the simple Lie algebra $A_3$, the fundamental domain of the affine Weyl group $F_Q$ is according to \eqref{FQ} of the form
\begin{equation*}
F_Q=\set{b_1\omega_1+\dots+b_3\omega_3}{b_0+ b_1+b_2+b_3=1,b_i\geq0,i=0,1,2,3},
\end{equation*}  
and the fundamental domain of the extended dual affine Weyl group $F_P\subset F_Q$ is by \eqref{lexorder} of the explicit form 
 \begin{equation*}
\begin{aligned}
F_P=&\set{b\in F_Q}{b_0>b_1,b_0>b_2,b_0>b_3}\cup\set{b\in F_Q}{b_0=b_1> b_2,b_0>b_3}\\&\cup\set{b\in F_Q}{b_0=b_2>b_1\geq b_3}
\cup\set{b\in F_Q}{b_0=b_1=b_2\geq b_3}.
\end{aligned}
\end{equation*}
The sets \eqref{FQS} coincide for the identity sign homomorphism with the original sets $F_Q$ and $F_P$, 
$$F_Q^\id=F_Q,\q  F_P^\id=F_P.$$
For the sign homomorphism $\si^e$, the set $F_Q^{\sigma^e}$ is by \eqref{FQS} and \eqref{fun1} of the form
\begin{equation*}
F_Q^{\sigma^e}=\set{b^{\sigma^e}_1\omega_1+b^{\sigma^e}_2\omega_2+b^{\sigma^e}_3\omega_3}{b^{\sigma^e}_0+b^{\sigma^e}_1+b^{\sigma^e}_2+b^{\sigma^e}_3=1,b^{\sigma^e}_i>0,i=0,1,2,3},
\end{equation*}
and relation \eqref{fundP} together with Table \ref{gamtab} implies the following explicit description of  $F_P^{\sigma^e},$
\begin{equation*}
\begin{aligned}
F_P^{\sigma^e}=&\set{b^{\sigma^e}\in F_Q^{\sigma^e}}{b^{\sigma^e}_0>b^{\sigma^e}_1,b^{\sigma^e}_0>b^{\sigma^e}_2,b^{\sigma^e}_0>b^{\sigma^e}_3}\cup\set{b^{\sigma^e}\in F^{\sigma^e}_Q}{b^{\sigma^e}_0=b^{\sigma^e}_1> b^{\sigma^e}_2,b^{\sigma^e}_0>b^{\sigma^e}_3}\\
&\cup\set{b^{\sigma^e}\in F^{\sigma^e}_Q}{b^{\sigma^e}_0=b^{\sigma^e}_2>b^{\sigma^e}_1\geq b^{\sigma^e}_3}\cup\set{b^{\sigma^e}\in F^{\sigma^e}_Q}{b^{\sigma^e}_0=b^{\sigma^e}_1=b^{\sigma^e}_2> b^{\sigma^e}_3}.
\end{aligned}
\end{equation*}
The domains $F_Q$, $F_P$, $F_Q^{\sigma^e}$, $F_P^{\sigma^e}$ and $WF_P$ of $A_3$ are depicted in Figure~\ref{obrFPA3}.
\begin{figure}
\resizebox{5cm}{!}{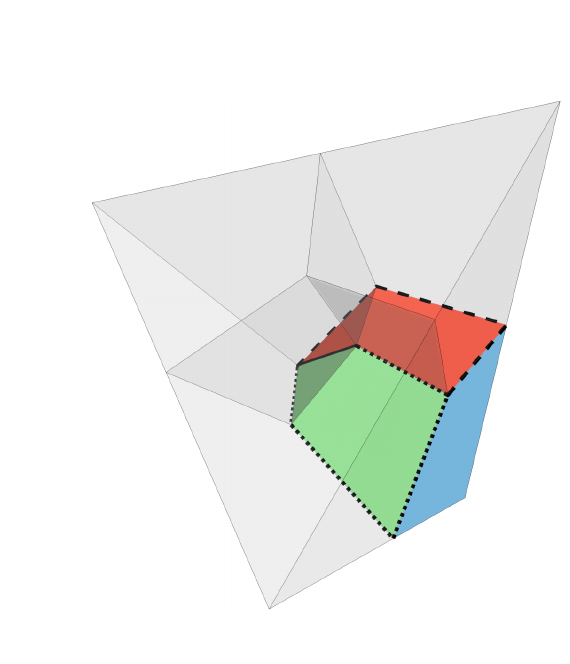}\hspace{1.8cm}\resizebox{5cm}{!}{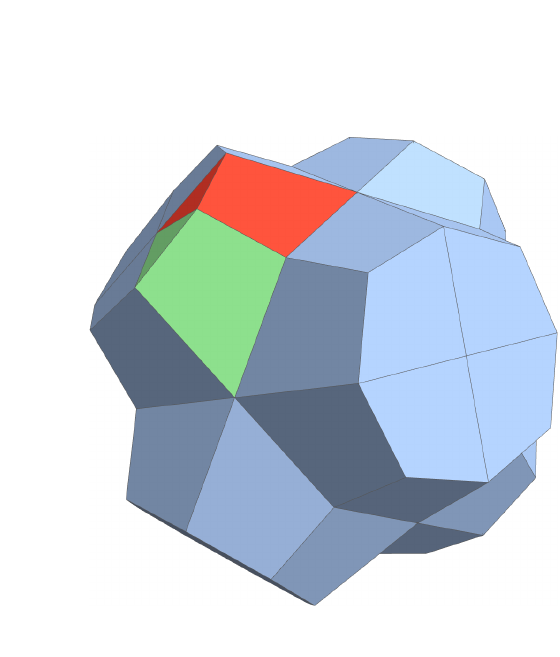}
\caption{(a) The domain $F_P$ of $A_3$, bounded by six kite-shaped surfaces, contains its entire boundary except for the points lying on the green surfaces and the dotted lines. The domain $F_P^{\sigma^e}$ is obtained from $F_P$ by omitting the blue areas together with the dashed lines and the point in the intersection of the top three kite-shaped surfaces. The enveloping tetrahedron represents the domain $F_Q$, tiled by $\Gamma-$images of $F_P$. (b) The polyhedron is formed by images of the Weyl group $W$ acting on $F_P$.}
\label{obrFPA3}
\end{figure}

Setting $x=x_1\omega_1+x_2\omega_2+x_3\omega_3,x_1,x_2,x_3\in\R^{\geq0}$, another geometric shape for the choice of the fundamental domain of the extended dual affine Weyl group exists. The domain $\mathcal{B}$, bounded by the planes orthogonal to $\omega_i$ and passing through $\omega_i/2$,
\begin{equation*}
\begin{aligned}
\mathcal{B}=&\set{x\in\R^3}{3x_1+2x_2+x_3<\tfrac{3}{2},x_1+2x_2+x_3<1,x_1+2x_2+3x_3<\tfrac32}\cup\\
&\set{x\in\R^3}{3x_1+2x_2+x_3=\tfrac{3}{2},x_1+2x_2+x_3\leq1,x_1+2x_2+3x_3<\tfrac32}\cup\\
&\set{x\in\R^3}{3x_1+2x_2+x_3<\tfrac{3}{2},x_1+2x_2+x_3=1,x_1+2x_2+3x_3<\tfrac32,x_1\geq x_3}\cup\\
&\set{x\in\R^3}{x_1=x_2=x_3=1/4},
\end{aligned}
\end{equation*}
is also a fundamental domain of $W_P^{\mathrm{aff}}$. The closure of the domain $W\mathcal{B}$ coincides, up to a factor $2\pi$, with the Brillouin zone \cite{Mich} of the dual root lattice of $A_3$. The domains $\mathcal{B}$ and $W\mathcal{B}$ are depicted in Figure \ref{mic}.
\begin{figure}
\resizebox{5cm}{!}{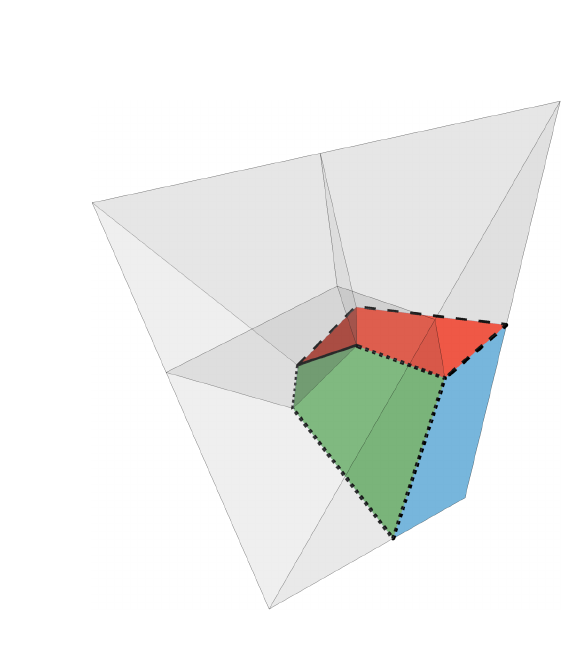}\hspace{1.8cm}\resizebox{5cm}{!}{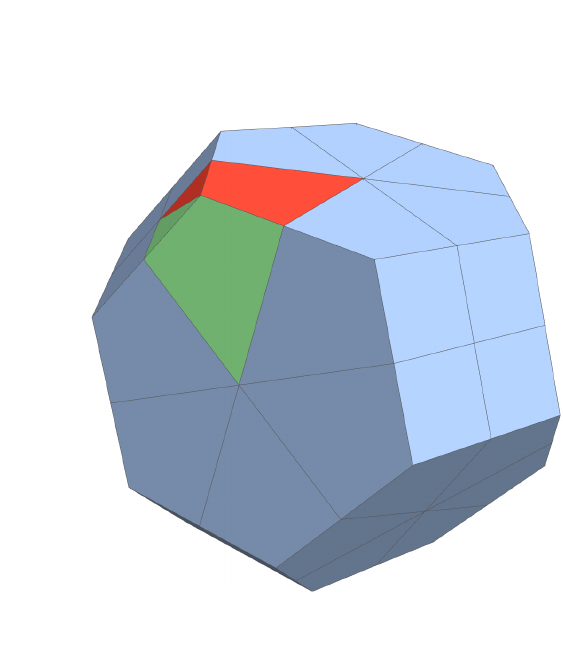}
\caption{(a) The domain $\mathcal{B}$ of $A_3$, bounded by one square-shaped and  five kite-shaped surfaces, contains its entire boundary except for the points on the green areas and the dotted lines. The enveloping tetrahedron represents the domain $F_Q$, tiled by $\Gamma-$images of $\mathcal{B}$. (b) The polyhedron is formed by the images of the Weyl group $W$ acting on $\mathcal{B}$.}
\label{mic}
\end{figure}
\end{example}

\subsection{$(\mathcal{R},\sigma)-$invariant polynomials}\label{invar}\

The vector space $\C[x]$ comprises polynomials of $n+1$ variables $ x\equiv [x_0,\dots,x_n]^T$ over $\C$. 
The extended $m-$degree $\mathrm{edg}_m\,x^\lambda$ of a monomial $x^\lambda \equiv x_0^{\lambda_0}x_1^{\lambda_1}\dots x_n^{\lambda_n}$, $\lambda\equiv[\la_0,\dots,\la_n]$ is defined as $$\mathrm{edg}_m\,x^\lambda \equiv \lambda_0+m_1^\vee\lambda_1+\dots+m_n^\vee\lambda_n.$$
The extended $m-$degree of any polynomial $f\in\C[x]$ is then the maximum extended $m-$degree of homogeneous parts of $f$. Recall from \cite{HP} that the finite sets $\Lambda_M \subset P $ consist for any $M \in \N$ of the weights contained in the set $MF_Q$, 
\begin{equation}\label{lambda0}
\Lambda_M=P \cap MF_Q,
\end{equation}
and is of the explicit form
\begin{equation}\label{lambda}
\Lambda_M=\setm{\la_1\omega_1+\dots+\lambda_n\omega_n}{\la_0+m_1^\vee\la_1+\dots+m_n^\vee\la_n=M,\,\la_i\in\Z^{\geq0},i=0,\dots,n}. 
\end{equation}
Identifying each element $\la$ of $\Lambda_M$ with its Kac coordinates $[\la_0,\dots,\la_n]$, results to conclusion that $\mathrm{edg}_m\,x^\lambda=M$ if and only if $\lambda\in\Lambda_M$.
All linear combinations of monomials of extended $m-$degree equal to $M$ form a vector subspace of $\C[x]$ denoted by $\Pi_M$,
$$\Pi_M\equiv\setm{\sum_{\lambda \in \Lambda_M} c_\lambda x^\lambda}{ c_\lambda\in\C, \, \lambda\in\Lambda_M} .$$
¨
The standard action \cite{invtheory,H2} of an operator $\mathbb{G}\in \mathrm{GL}_{n+1}(\C)$ on $\C[x]$ is given by
\begin{equation}\label{matakce}
\mathbb{G}\cdot f(x)\equiv f(\mathbb{G}^{-1}x),\q f\in\C[x].
\end{equation}
For any representation $\mathcal{R}: \Gamma \map \mathrm{GL}_{n+1}(\C)$ of the abelian group \eqref{gamma} and a sign homomorphism $\si$, a polynomial $f\in\C[x]$ is called $(\mathcal{R},\sigma)-$invariant if it for all $\gamma\in\Gamma$ satisfies 
\begin{equation}\label{sigakce}
\mathcal{R}(\gamma)\cdot f=\sigma\circ\widehat\psi(\gamma)\,f.
\end{equation} 
A vector subspace $\Pi_M^{\mathcal{R},\sigma} \subset \Pi_M$ contains all $(\mathcal{R},\sigma)-$invariant polynomials from $\Pi_M$,
$$\Pi_M^{\mathcal{R},\sigma}\equiv \setb{f\in\Pi_M}{\mathcal{R}(\gamma)\cdot f=\sigma\circ\widehat\psi(\gamma)\,f, \,\, \gamma\in \Gamma}.$$
\begin{tvr}\label{iso}
Let $\mathcal{R}_1,\mathcal{R}_2$ be representations of $\Gamma$ to  $\mathrm{GL}_{n+1}(\C)$ for which there exists $\mathbb{P}\in\mathrm{GL}_{n+1}(\C)$ such that 
\begin{enumerate}[(i)]
\item $\mathcal{R}_2(\gamma)=\mathbb{P}^{-1}\mathcal{R}_1(\gamma)\mathbb{P}$ for all $\gamma\in \Gamma$, i.e. $\mathcal{R}_1$ and $\mathcal{R}_2$ are equivalent,
\item $\mathbb{P}\cdot f\in\Pi_M$ $\Leftrightarrow$ $f\in \Pi_M$.
\end{enumerate}
Then the spaces $\Pi^{\mathcal{R}_1,\sigma}_M$ and $\Pi^{\mathcal{R}_2,\sigma}_M$ are for any sign homomorphism $\sigma$ isomorphic, 
$$\Pi^{\mathcal{R}_1,\sigma}_M\simeq\Pi^{\mathcal{R}_2,\sigma}_M.$$ 
\end{tvr}
\begin{proof}
Assumption (ii) implies for any $f\in\Pi_M^{\mathcal{R}_2,\sigma}$ that $\mathbb{P}\cdot f \in \Pi_M$.  Definitions \eqref{matakce} and \eqref{sigakce} and assumption (i) yield $(\mathcal{R}_1,\sigma)-$invariance of $\mathbb{P}\cdot f,$
$$\mathcal{R}_1(\gamma)\cdot\mathbb{P}\cdot f(x)=f(\mathbb{P}^{-1}\mathcal{R}_1(\gamma)^{-1}x)=f(\mathcal{R}_2(\gamma)^{-1}
\mathbb{P}^{-1}x)=\mathbb{P}\cdot \mathcal{R}_2(\gamma)\cdot f(x)=\sigma\circ\widehat\psi(\gamma)\,\mathbb{P}\cdot f(x)$$
and therefore $\mathbb{P}\cdot f\in \Pi_M^{\mathcal{R}_1,\sigma}$. Conversely, it holds that if $f\in\Pi_M^{\mathcal{R}_1,\sigma}$ then $\mathbb{P}^{-1}\cdot f\in \Pi_M^{\mathcal{R}_2,\sigma}$. The map $\Pi_M^{\mathcal{R}_2,\sigma}\ni f \mapsto \mathbb{P}\cdot f\in\Pi_M^{\mathcal{R}_1,\sigma}$ is linear and its inverse is the map $\Pi_M^{\mathcal{R}_1,\sigma}\ni f\mapsto \mathbb{P}^{-1}\cdot f\in\Pi_M^{\mathcal{R}_2,\sigma}.$
\end{proof}

\begin{tvr}\label{cond}
Let $\mathbb{P}\in\mathrm{GL}_{n+1}(\C)$ be such that 
$(\mathbb{P}^{-1}x)_i $ is a linear combination of monomials of the extended $m-$degree $m_i^\vee$ for every $i\in \{0,1,\dots,n\}$.
Then $f\in\Pi_M$ if and only if $ \mathbb{P}\cdot f\in\Pi_M$.
\end{tvr}
\begin{proof}
For any monomial $x^\lambda\in\Pi_M$, the assumption implies that the factors of $\mathbb{P}\cdot x^\lambda$,
\begin{equation}\label{Pprod}
\mathbb{P}\cdot x^\lambda=(\mathbb{P}^{-1}x)_0^{\la_0}(\mathbb{P}^{-1}x)_1^{\la_1}\dots(\mathbb{P}^{-1}x)_n^{\la_n}, 	
\end{equation}
satisfy for every $i\in \{0,1,\dots,n\}$ that
\begin{equation}\label{Pf}
(\mathbb{P}^{-1}x)_i^{\lambda_i}=\left(\sum_{k=1}^l p_{j_k} x_{j_k}\right)^{\lambda_i},\q p_{j_k}\in\C\setminus \{0\},\q \mathrm{edg}_m\, x_{j_k}=m_i^\vee,\, k=1,\dots,l.	
\end{equation}
Multinomial expansion of \eqref{Pf} provides the relation
\begin{equation*}
(\mathbb{P}^{-1}x)_i^{\lambda_i} = \sum_{\begin{smallmatrix}r_1, r_2, \dots, r_l \geq 0 \\ r_1 + r_2+ \dots + r_l = \la_i  \end{smallmatrix}}\left( \frac{\la_i !}{r_1! r_2!\dots r_l !}
p_{j_1}^{r_1}p_{j_2}^{r_2}\dots p_{j_l}^{r_l} \right) x_{j_1}^{r_1}x_{j_2}^{r_2}\dots x_{j_l}^{r_l},
\end{equation*}
which guarantees that the polynomial $(\mathbb{P}^{-1}x)_i^{\lambda_i}$ is a linear combination of monomials of extended $m-$degree $\lambda_i m_i^\vee$ and thus
$$\mathrm{edg}_m\,(\mathbb{P}^{-1}x)_i^{\lambda_i}=\lambda_i m_i^\vee .$$
Since $\mathbb{P}\cdot x^\lambda$ is a product of non-zero polynomials  $(\mathbb{P}^{-1}x)_i^{\lambda_i}$ in \eqref{Pprod}, it is also non-zero and a linear combination of monomials of extended $m-$degree $M$,
$$\mathrm{edg}_m\,\mathbb{P}\cdot x^\lambda=\mathrm{edg}_m\,x^\lambda.$$
The resulting polynomial $\mathbb{P}\cdot f$ is thus a linear combination of monomials of extended $m-$degree $M$ and therefore $ \mathbb{P}\cdot f\in\Pi_M$.
Conversely, the assumption on $\mathbb{P}$ is equivalent to $\mathbb{P}^{-1}$ being a block diagonal matrix in a suitable permutation of the standard basis with its blocks determined by the same values of~$m_i^\vee$.  Taking into account that the matrix $\mathbb{P}$ retains the same form, $f\in \Pi_M$ forces also $\mathbb{P}^{-1}\cdot f \in\Pi_M$.  
\end{proof}

The action \eqref{akcegamma} of $\Gamma$ on $F_Q$ via permutations of the Kac coordinates $[b_0,\dots,b_n]$ induces a faithful representation $\mathcal{A}: \Gamma \map \mathrm{GL}_{n+1}(\C)$ by assigning each element $\gamma\in\Gamma$ its permutation matrix
$\mathcal{A} (\gamma)\in \mathrm{GL}_{n+1}(\C)$, 
\begin{equation}\label{permdef}
\mathcal{A} (\gamma)[b_0,\dots,b_n]^T=[b_{\pi_\gamma(0)},\dots,b_{\pi_\gamma(n)}]^T.	
\end{equation}
Similarly, the action of $\Gamma_M$ on $MF_Q$, described by \eqref{akcegamM}, assigns to each element $\gamma_M\in\Gamma_M$ its permutation matrix corresponding to $\gamma\in\Gamma$, provided by isomorphism \eqref{izo}. Consequently, the action of $\Gamma_M$ also induces, independently on $M$, the identical representation matrices $\mathcal{A}(\gamma)$,
\begin{equation}\label{matice}
\mathcal{A}(\gamma)=\mathcal{A}(\gamma_M).
\end{equation}

Since the finite group $\Gamma$ is abelian, the commuting diagonalizable matrices $\mathcal{A}(\gamma)$ can be simultaneously diagonalized, i.e. there exists a unitary matrix $\mathbb{P}\in\mathrm{GL}_{n+1}(\C)$ such that $\mathbb{P}^{\dagger}\mathcal{A}(\gamma)\mathbb{P}$ is a diagonal matrix for each $\gamma\in\Gamma$. Therefore, the diagonal representation $\mathcal{D}: \Gamma \map \mathrm{GL}_{n+1}(\C)$, equivalent to $\mathcal{A}$, is given by 
\begin{equation}\label{diagdef}
\mathcal{D}(\gamma)=\mathbb{P}^{\dagger}\mathcal{A}(\gamma)\mathbb{P}.	
\end{equation}
The diagonal matrices $\mathcal{D}(\gamma)$ are for the generators of the non-trivial groups $\Gamma$ of the following form 
\begin{equation}\label{diagonal}
\begin{alignedat}{2}
&A_n:&\q & \mathcal{D}(\gamma_1)=\mathrm{diag}\left(1,e^{\frac{2\pi\i}{n+1}},e^{\frac{4\pi\i}{n+1}},\dots,e^{\frac{2n\pi\i}{n+1}}\right),\\
&B_{2k+1}:&\q &\mathcal{D}(\gamma_{2k+1})=\mathrm{diag}\left(1,-1,1,-1,\dots,1,-1\right),\\
&B_{2k}:&\q&\mathcal{D}(\gamma_{2k})=\mathrm{diag}\left(1,1,1,-1,1,-1\dots,1,-1,-1\right),\\
&C_n:&\q &\mathcal{D}(\gamma_1)=\mathrm{diag}\left(1,-1,1,1,\dots,1,1\right),\\
&D_{2k+1}:&\q &\mathcal{D}(\gamma_{2k+1})=\mathrm{diag}\left(1,-1,1,-1,\dots,1,-1,\i,-\i\right),\\
&D_{2k}:&\q & \mathcal{D}(\gamma_1)=\mathrm{diag}\left(1,1,\dots,1,-1,-1\right),\\
&&\q &\mathcal{D}(\gamma_{2k})=\mathrm{diag}\left(1,-1,1,-1,\dots,1,-1,1\right),\\
&E_6:&\q &\mathcal{D}(\gamma_5)=\mathrm{diag}\left(1,e^{\frac{2\pi\i}{3}},e^{\frac{4\pi\i}{3}},1,e^{\frac{2\pi\i}{3}},e^{\frac{4\pi\i}{3}},1\right),\\
&E_7:&\q & \mathcal{D}(\gamma_6)=\mathrm{diag}\left(1,1,1,1,-1,1,-1,-1\right)
\end{alignedat}
\end{equation}
and the unitary conjugation matrices $\mathbb{P}$ are listed in Table~\ref{matP}.

\begin{thm}\label{isoAD}
Let $\mathcal{A}$ be the permutation representation \eqref{permdef} of $\Gamma$ and $\mathcal{D}$ the corresponding diagonal representation \eqref{diagdef}. Then the spaces $\Pi^{\mathcal{A},\sigma}_M$ and $\Pi^{\mathcal{D},\sigma}_M$ are for any sign homomorphism $\sigma$ isomorphic, 
\begin{equation*}
\Pi^{\mathcal{A},\sigma}_M\simeq \Pi^{\mathcal{D},\sigma}_M.
\end{equation*}
\end{thm}
\begin{proof}
Explicit forms of the unitary conjugation matrices $\mathbb{P}$ in Table \ref{matP} guarantee that each $\mathbb{P}$ satisfies the assumption of Proposition~\ref{cond}. 
\begin{table}[!ht]
{\small
\begin{tabular}{|c||l|c||l|}
\hline
&$\mathbb{P}=(p_{ij})\in\C^{n+1,n+1},\q i,j=0,\dots,n$&&$\mathbb{P}=(p_{ij})\in\C^{n+1,n+1},\q i,j=0,\dots,n$\\
\hline\hline
\multirow{2}{*}{$A_n$}&$p_{ij}=\frac{1}{\sqrt{n+1}}\alpha^{-ij},\q i,j=0,\dots,n,\q \alpha=e^{\frac{2\pi \i}{n+1}}$&\multirow{2}{*}{$C_n$}&$p_{00}=p_{10}=p_{01}=\frac{1}{\sqrt{2}},\q p_{11}=-\frac{1}{\sqrt{2}}$,\\
&&&$p_{ii}=1,\q i=2,\dots,n$\\\hline
\multirow{3}{*}{$B_{2k+1}$}&$p_{ii}=p_{i,2k+1-i}=\frac{1}{\sqrt{2}},$&\multirow{3}{*}{$B_{2k}$}&$p_{00}=p_{2k,0}=p_{0,2k}=\frac{1}{\sqrt{2}}$, \\
&$p_{2k+1-i,i}=\frac{(-1)^i}{\sqrt{2}},$&&$p_{k,1}=1,\q p_{2k,2k}=-\frac{1}{\sqrt{2}}$,\\
&$p_{2k+1-i,2k+1-i}=\frac{(-1)^{i+1}}{\sqrt{2}},\q i=0,\dots, k$&&$p_{i-1,i}=p_{i-1,2k+1-i}=\frac{1}{\sqrt{2}},$\\
&&&$ p_{2k+1-i,i}=\frac{(-1)^{i}}{\sqrt{2}}$,\\
&&&$ p_{2k+1-i,2k+1-i}=\frac{(-1)^{i+1}}{\sqrt{2}},\q i=2,\dots, k$\\\hline
\multirow{6}{*}{$D_{2k+1}$}&$p_{00}=p_{01}=p_{10}=p_{11}=\frac12$,&\multirow{6}{*}{$D_{2k}$}&$p_{00}=p_{01}=p_{10}=p_{11}=p_{2k-1,0}=p_{2k,0}=\frac12$,\\
&$p_{2k,0}=p_{2k+1,0}=p_{2k+1,2k}=p_{2k+1,2k+1}=\frac12$,&&$p_{0,2k-1}=p_{0,2k}=p_{2k-1,2k-1}=p_{2k,2k}=\frac12$\\
&$p_{2k,1}=p_{2k+1,1}=p_{2k,2k}=p_{2k,2k+1}=-\frac{1}{2}$,&&$p_{k,2}=1,\q p_{2k-1,1}=p_{2k,1}=p_{1,2k-1}=-\frac12$,\\
&$p_{1,2k}=p_{0,2k+1}=-\frac{\i}{2},\q p_{0,2k}=p_{1,2k+1}=\frac{\i}{2}$,&&$p_{1,2k}=p_{2k-1,2k}=p_{2k,2k-1}=-\frac12$,\\
&$p_{i,2i-2}=p_{i,2i-1}=p_{2k+1-i,2i-2}=\frac{1}{\sqrt{2}}$,&&$p_{i-1,i}=p_{i-1,2k+1-i}=\frac{1}{\sqrt{2}},\q p_{2k+1-i,i}=\frac{(-1)^i}{\sqrt{2}}$,\\
&$p_{2k+1-i,2i-1}=-\frac{1}{\sqrt{2}},\q i=2,\dots,k$&&$p_{2k+1-i,2k+1-i}=\frac{(-1)^{i+1}}{\sqrt{2}},\q i=3,\dots,k$\\\hline
$E_6$&$\frac{1}{\sqrt{3}}\begin{pmatrix}
1&e^{\frac{4\pi\i}{3}}&0&0&0&e^{\frac{2\pi\i}{3}}&0\\
1&e^{\frac{2\pi\i}{3}}&0&0&0&e^{\frac{4\pi\i}{3}}&0\\
0&0&e^{\frac{2\pi\i}{3}}&0&e^{\frac{4\pi\i}{3}}&0&1\\
0&0&0&\sqrt{3}&0&0&0\\
0&0&e^{\frac{4\pi\i}{3}}&0&e^{\frac{2\pi\i}{3}}&0&1\\
1&1&0&0&0&1&0\\
0&0&1&0&1&0&1
\end{pmatrix}$&$E_7$&$\frac{1}{\sqrt{2}}\begin{pmatrix}
1&0&0&0&0&0&1&0\\
0&1&0&0&0&0&0&1\\
0&0&1&0&1&0&0&0\\
0&0&0&\sqrt{2}&0&0&0&0\\
0&0&1&0&-1&0&0&0\\
0&1&0&0&0&0&0&-1\\
1&0&0&0&0&0&-1&0\\
0&0&0&0&0&\sqrt{2}&0&0
\end{pmatrix}$\\\hline
\end{tabular}}
\bigskip
\caption{The non-zero entries of the unitary conjugation matrices $\mathbb{P}$, which for non-trivial groups $\Gamma$ diagonalize their permutation representations $\mathcal{A}$.}
\label{matP}
\end{table} 
This in turn, results to the validity of assumption (ii) in Proposition~\ref{iso}.
\end{proof}

\section{Sets of points and weights}

\subsection{Sets $\Lambda^\sigma_{P,M}$ and $F^\sigma_{Q^\vee,M}$}\

The Fourier calculus of Weyl orbit functions is commonly formulated \cite{HMPdis,HP,diejen} on points $F_{P^\vee,M}^\sigma$, $M\in\N$ from the refined dual weight lattice $\tfrac{1}{M}P^\vee$ contained in $F^\sigma$,
$$F_{P^\vee,M}^\sigma=\tfrac{1}{M}P^\vee\cap F^\sigma,$$
while weights $\Lambda^\sigma_{Q,M}$ labelling the functions are taken from the magnified dual fundamental domain $MF_Q^\sigma$,
\begin{equation*}
\Lambda^\sigma_{Q,M}=P\cap MF_Q^\sigma .	
\end{equation*}
This work develops the Fourier calculus on points $F^\sigma_{Q^\vee,M}$, $M>m^\sigma$ from the refined dual root lattice $\tfrac{1}{M}Q^\vee$ contained in $F^\sigma$,
\begin{equation}\label{FQvee}
F^\sigma_{Q^\vee,M}\equiv \frac{1}{M}Q^\vee \cap F^\sigma,
\end{equation}  
with weights $\Lambda^\sigma_{P,M}$ labelling the functions taken from the magnified fundamental domain $MF^\sigma_P$,
\begin{equation}\label{LP}
\Lambda^\sigma_{P,M}\equiv P \cap MF^\sigma_P.
\end{equation}
The set of points $F^\sigma_{Q^\vee,M}\subset F_{P^\vee,M}^\sigma$ contains all points of  $F_{P^\vee,M}^\sigma$ belonging to $\tfrac{1}{M}Q^\vee$
and the set of weights $\Lambda^\sigma_{P,M}\subset \Lambda^\sigma_{Q,M}$  is a suitable fragment of $ \Lambda^\sigma_{Q,M}$. The point sets $F_{Q^\vee,7}^\sigma$ and $F_{P^\vee,7}^\sigma$ together with the weight sets $\Lambda_{P,7}^\sigma$ and $\Lambda_{Q,7}^\sigma$ of  $A_2$ are depicted in  Figure~\ref{A2body}.

\begin{figure}[!ht]
\resizebox{5.2cm}{!}{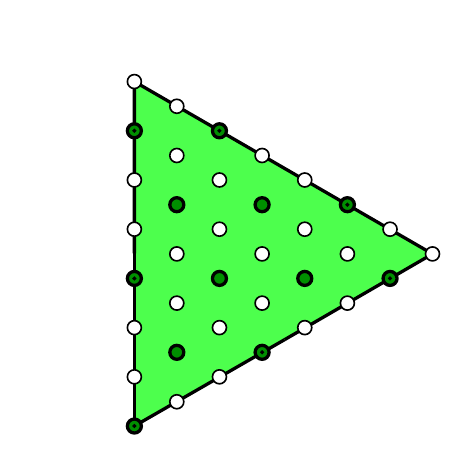}\hspace{1.7cm}\resizebox{5.2cm}{!}{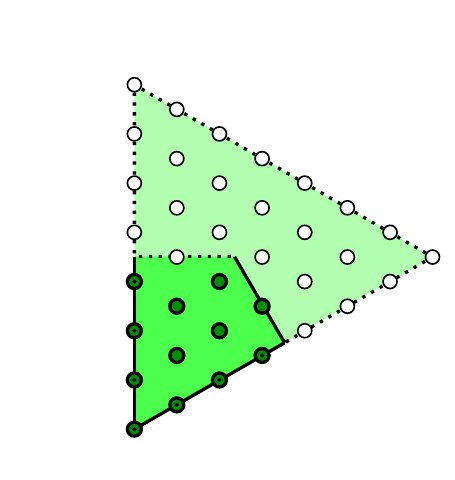}
\caption{$(a)$ The fundamental domain $F$ of $A_2$, depicted as the green triangle, contains 24 white nodes and 12 green nodes forming together the set of points $F_{P^\vee,7}^\id$. The 12 green nodes represent the points of $F_{Q^\vee,7}^\id$ .
The green nodes without the dotted ones on the boundary of $F$ depict the 5 points of $F_{Q^\vee,7}^{\sigma^e}$. 
$(b)$ The domain $7F_Q$ of $A_2$, depicted as the light green triangle, contains 24 white nodes and 12 green nodes that form the set of weights $\Lambda_{Q,7}^\id$. The 12 green nodes in the kite-shaped domain $7F_P$ belong to the weight set $\Lambda_{P,7}^\id$. The green nodes without the dotted ones depict the 5 points of the set $\Lambda_{P,7}^{\sigma^e}$.} 
\label{A2body}
\end{figure}

In order to describe explicitly the point sets $F^\sigma_{Q^\vee,M}$ for any sign homomorphism $\sigma$, the symbols $s^\sigma_i$,  $i\in \{0,1,\dots,n\}$ are introduced via the relations
\begin{equation*}
\begin{alignedat}{4}
&s^\id_0,s^\id_i\in\Z^{\geq 0} &\q &\text{ if } \alpha_i\in\Delta,\\
&s^{\sigma^e}_0,s^{\sigma^e}_i\in\N &\q &\text{ if } \alpha_i\in\Delta,\\
&s^{\sigma^s}_0,s^{\sigma^s}_i\in\Z^{\geq 0} &\q &\text{ if } \alpha_i\in\Delta_l,&\q &s^{\sigma^s}_i\in\N &\q &\text{ if }\al_i\in\Delta_s,\\
&s^{\sigma^l}_0,s^{\sigma^l}_i\in\N& \q&\text{ if } \alpha_i\in\Delta_l,&\q& s^{\sigma^l}_i\in\Z^{\geq 0} &\q&\text{ if }\al_i\in\Delta_s.
\end{alignedat}
\end{equation*} 
The explicit description of the point set $F^\sigma_{Q^\vee,M}$ follows from relations \eqref{alom} and \eqref{domainF},
\begin{equation}\label{fundQvee}
F^\sigma_{Q^\vee,M}=\setb{\frac{s^\sigma_1}{M}\omega_1^\vee+\dots+\frac{s^\sigma_n}{M}\omega_n^\vee}{\sum_{i=0}^n m_i s^\sigma_i=M,\, \sum_{i=1}^n(C^T)^{-1}_{ij}s^\sigma_i\in\Z,\, j=1,\dots,n}
\end{equation}
where $C^T$ denotes the transposed Cartan matrix \eqref{Cart}.
The set of $n$ equations 
\begin{equation}\label{condpre2}
\sum_{i=1}^n(C^T)^{-1}_{ij}s^\sigma_i\in\Z,\q j=1,\dots,n	
\end{equation}
selects such points from $F_{P^\vee,M}^\sigma$ which belong to $\tfrac{1}{M}Q^\vee$. Standard reduction of the sets of equations \eqref{condpre2} yields their simplified form for each non-trivial case as
\begin{equation}\label{cond2}
\begin{alignedat}{2}
&A_n:&\q & s^\sigma_1+2s^\sigma_2+3s_3^\sigma+\dots+ ns^\sigma_n\equiv 0 \mod {n+1},\\
&B_{2k+1}:&\q & s^\sigma_1+s^\sigma_3+s^\sigma_5+\dots+s^\sigma_{2k+1}\equiv 0 \mod 2 ,\\
&B_{2k}:&\q & s^\sigma_1+s^\sigma_3+s^\sigma_5+\dots+s^\sigma_{2k-1}\equiv 0 \mod 2 ,\\
&C_n:&\q & s^\sigma_n\equiv 0 \mod 2,\\
&D_{4k:}&\q & s^\sigma_1+s^\sigma_3+\dots+s^\sigma_{4k-3}+s^\sigma_{4k-1}\equiv 0 \mod 2,\\
&&& s^\sigma_{4k-1}+s^\sigma_{4k}\equiv 0 \mod 2,\\
&D_{4k+1}:&\q & 2s^\sigma_1+2s^\sigma_3+\dots+ 2s^\sigma_{4k-1}+3s^\sigma_{4k}+s^\sigma_{4k+1} \equiv 0 \mod 4,\\
&D_{4k+2}:&\q &s^\sigma_1+s^\sigma_3+\dots+s^\sigma_{4k-1}+s^\sigma_{4k+2}\equiv 0 \mod 2,\\
&&& s^\sigma_{4k+1}+s^\sigma_{4k+2}\equiv 0 \mod 2,\\
&D_{4k+3}:&\q &2s^\sigma_1+2s^\sigma_3+\dots+ 2s^\sigma_{4k+1}+s^\sigma_{4k+2}+3s^\sigma_{4k+3} \equiv 0 \mod 4,\\
&E_6:&\q & s^\sigma_1+2s^\sigma_2+s^\sigma_4+2s^\sigma_5\equiv 0 \mod 3,\\
&E_7:&\q & s^\sigma_4+s^\sigma_6+s^\sigma_7\equiv 0 \mod 2.
\end{alignedat}
\end{equation}
 
\subsection{Cardinality of $\Lambda^\sigma_{P,M}$ and $F^\sigma_{Q^\vee,M}$}\

In order to prove that the point sets $F^\sigma_{Q^\vee,M}$ and the sets of weights $\Lambda^\sigma_{P,M}$ have for $M> m^\sigma$ the same cardinality, the isomorphism of polynomial spaces in Theorem \ref{isoAD} is employed. The first step is to introduce an auxiliary finite set of weights $\wt\Lambda^\sigma_{Q,M}$ by the relation
\begin{equation}\label{wLQM}
\wt\Lambda^\sigma_{Q,M}\equiv \setb{\la \in \Lambda_M}{\sigma\circ\widehat\psi(\mathrm{Stab}_{\Gamma_M}(\la))=\{1\}}
\end{equation}
together with its complementary set $\wt{H}_{Q,M}^\sigma$,
\begin{equation}\label{wHQM}
\wt{H}_{Q,M}^\sigma\equiv\set{\la\in \Lambda_M }{\sigma\circ\widehat\psi(\mathrm{Stab}_{\Gamma_M}(\la))=\{\pm1\}}
\end{equation}
and their corresponding sets $\wt\Lambda^\sigma_{P,M}$ and  $\wt{H}_{P,M}^\sigma$ of representative weights in $\Gamma_M-$orbits,
\begin{align}\label{pomlam}
\wt\Lambda^\sigma_{P,M}  & \equiv  MF_P \cap  \wt\Lambda^\sigma_{Q,M},  \\
\wt{H}_{P,M}^\sigma  &  \equiv  MF_P \cap  \wt H^\sigma_{Q,M}.   \label{pomlamH}
\end{align}
Indeed, since $MF_P$ is a fundamental domain of the action of $\Gamma_M$ on $MF_Q$, it holds that 
\begin{align}
\wt\Lambda^\sigma_{Q,M}  &= \Gamma_M \wt\Lambda^\sigma_{P,M}, \label{gaml1} \\
\wt{H}_{Q,M}^\sigma  & =  \Gamma_M   \wt H^\sigma_{P,M}.  \label{gaml2}
\end{align}
The following disjoint decomposition of the weight set \eqref{lambda0} and the set $\Lambda^\id_{P,M}$ are thus obtained,
\begin{align}\label{disLam}
\Lambda_M &= \wt\Lambda^\sigma_{Q,M} \cup \wt{H}_{Q,M}^\sigma, \\
\Lambda^\id_{P,M} & = \wt\Lambda^\sigma_{P,M} \cup \wt{H}_{P,M}^\sigma. \label{disPLam}
\end{align}
The set $\wt\Lambda^\sigma_{P,M}$ is related by $\rho^\sigma$ shifts \eqref{rhoi} to the set $\Lambda^\sigma_{P,M+m^\sigma}$ in the following proposition.
\begin{tvr}\label{body}
The set of weights $\Lambda^\sigma_{P,M+m^\sigma}$ coincides for any $M\in\N$ with the $\rho^\sigma-$shifted set $\wt\Lambda^\sigma_{P,M},$
\begin{equation*}
\Lambda^\sigma_{P,M+m^\sigma}= \rho^\sigma+\wt\Lambda^\sigma_{P,M}.
\end{equation*}
\end{tvr}
\begin{proof}
Firstly, for any coordinates $[\la_0,\dots,\la_n]$ of $\la=\la_1\omega_1+\dots+\la_n\omega_n\in P$ are the corresponding coordinates $[\la^\sigma_0,\dots,\la^\sigma_n]$ of $\la^\sigma=\la^\sigma_1\omega_1+\dots+\la^\sigma_n\omega_n \in P$ defined by $$\la_i^\sigma=\la_i+\rho_i^\sigma ,\q i\in \{0,1,\dots, n\}.$$
According to \eqref{LP} and \eqref{MFP}, any weight $\la^\sigma \in \Lambda^\sigma_{P,M+m^\sigma}$ satisfies the conditions
\begin{align}\label{shifti}
\la^\sigma &\in P\cap(M+m^\sigma)F_P, \\
\la^\sigma &\in P\cap(M+m^\sigma) F^\sigma_Q,		\label{shifti2}
\end{align}
together with
\begin{equation}\label{shiftii}
\sigma\circ\widehat{\psi}\left(\mathrm{Stab}_{\Gamma_{M+m^\sigma}}\left(\la^\sigma\right)\right)=\{1\}.
\end{equation}
Defining relations \eqref{rhoi}, \eqref{msigma} and \eqref{FQS}, \eqref{fun1} imply that $\rho^\sigma \in P\cap m^\sigma F^\sigma_Q$ and consequently \eqref{shifti2} grants that 
\begin{equation}\label{inFQ}
\la\in P\cap MF_Q.
\end{equation}
Taking into account \eqref{gammarho} and comparing the action of $\gamma_{M+m^\sigma}\in \Gamma_{M+m^\sigma}$ on $\la^\sigma\in P\cap (M+m^\sigma)F_Q^\sigma$,
\begin{equation}\label{acti}
\gamma_{M+m^\sigma}\cdot \la^\sigma=[\la^\sigma_{\pi_\gamma(0)},\dots,\la^\sigma_{\pi_\gamma(n)}]=[\la_{\pi_\gamma(0)}+\rho_0,\dots,\la_{\pi_\gamma(n)}+\rho_n],	
\end{equation}
to the action of $\gamma_M \in \Gamma_{M}$ on $\la\in P\cap MF_Q$, 
\begin{equation}\label{actii}
\gamma_M\cdot\la=[\la_{\pi_\gamma(0)},\dots,\la_{\pi_\gamma(n)}],
\end{equation}
yields that the point $\la^\si$ is the lexicographically highest in its $\Gamma_{M+m^\sigma}-$orbit whenever $\la$ is the lexicographically highest in its $\Gamma_{M}-$orbit. Thus, properties \eqref{shifti}, \eqref{inFQ} and \eqref{lexorder} ensure that 
\begin{equation}\label{MFP2}
\la\in P\cap MF_P.	
\end{equation}
Connecting action \eqref{acti} for the stabilizing 
$\gamma_{M+m^\sigma}\in\mathrm{Stab}_{\Gamma_{M+m^\sigma}}\left(\la^\sigma\right)$
and action \eqref{actii} of $\gamma_M \in \Gamma_{M}$  on $\la\in P\cap MF_P$ forces $\gamma_M$ to stabilize $\la$, i.e.
$\gamma_M\in\mathrm{Stab}_{\Gamma_{M}}\left(\la\right).$
Relation \eqref{hatpsi}  thus gives 
\begin{equation}\label{sistab}
	\sigma\circ\widehat{\psi}\left(\mathrm{Stab}_{\Gamma_{M}}\left(\la\right)\right)=  \sigma\circ\widehat{\psi}\left(\mathrm{Stab}_{\Gamma_{M+m^\sigma}}\left(\la^\sigma\right)\right),
\end{equation}
moreover \eqref{shiftii} guarantees in turn that
\begin{equation}\label{stabGM}
\sigma\circ\widehat{\psi}\left(\mathrm{Stab}_{\Gamma_{M}}\left(\la\right)\right)=\{1\},
\end{equation}
and therefore $\la \in\wt\Lambda^\sigma_{P,M}.$  

Conversely, $\la \in\wt\Lambda^\sigma_{P,M}$ implies validity of \eqref{stabGM} and \eqref{MFP2}. Defining relations \eqref{rhoi}, \eqref{msigma} and \eqref{FQS}, \eqref{fun1} then imply that $\la\in P\cap MF_Q$ grants \eqref{shifti2}. By comparing actions \eqref{acti} and \eqref{actii}, property \eqref{MFP2} forces relation \eqref{shifti}, properties \eqref{stabGM} and \eqref{sistab} guarantee \eqref{shiftii} and therefore $\la^\sigma \in \Lambda^\sigma_{P,M+m^\sigma}$. 
\end{proof}

\begin{tvr}\label{lemA}
The dimension of $\Pi^{\mathcal{A},\sigma}_{M}$ is for any $M\in\N$ equal to the number of points in $\Lambda^\sigma_{P,M+m^\sigma}$,
$$\dim{\Pi^{\mathcal{A},\sigma}_{M}}=\abs{\Lambda^\sigma_{P,M+m^\sigma}}.$$
\end{tvr}
\begin{proof}
Utilizing the group action \eqref{akcegamM}, the polynomial $f_{\la}^\sigma\in \Pi_M$ is for any sign homomorphism $\sigma$ and any $\lambda\in\wt\Lambda^\sigma_{P,M}$ introduced by 
\begin{equation}\label{fla}
f_{\la}^\sigma(x)\equiv\sum_{\gamma\in \Gamma_M}\sigma\circ\widehat\psi (\gamma)\, x^{\gamma\cdot\lambda}.	
\end{equation}
Any $\gamma\in\Gamma_M$ is represented by \eqref{matice} as a permutation matrix $\mathcal{A}(\gamma)$ and thus the polynomial action \eqref{matakce} on the monomials takes the form 
\begin{equation}\label{Ax}
\mathcal{A}(\gamma)\cdot x^\lambda=x^{\gamma \cdot\lambda}.	
\end{equation}
Substituting \eqref{Ax} into \eqref{fla} and taking into account relations \eqref{hatpsi} and \eqref{matice},   
the polynomials $f^\sigma_\lambda$, $\lambda\in\wt\Lambda^\sigma_{P,M}$ are deduced to be $(\mathcal{A},\sigma)-$invariant, $$f^\sigma_\lambda\in\Pi_M^{\mathcal{A},\sigma}.$$
Since any $\gamma\in\mathrm{Stab}_{\Gamma_M} \left(\la\right)$ satisfies from \eqref{pomlam} that $\sigma\circ\widehat\psi(\gamma)=1$, a general linear combination of polynomials \eqref{fla} is of the form
\begin{align}\nonumber
\sum_{\lambda\in\wt\Lambda_{P,M}^\sigma}c_\lambda f_\lambda^\sigma(x)&=\sum_{\lambda\in\wt\Lambda_{P,M}^\sigma}\sum_{\gamma\in\Gamma_M}c_\lambda\, \sigma\circ\widehat\psi(\gamma)\,x^{\gamma\cdot\lambda}\\ \nonumber
&=\sum_{\lambda\in\wt\Lambda_{P,M}^\sigma}\left(\sum_{\gamma\in\mathrm{Stab}_{\Gamma_M} \left(\lambda\right)}c_\lambda\,x^\lambda+\sum_{\gamma\in\Gamma_M\setminus
\mathrm{Stab}_{\Gamma_M}\left(\lambda\right)}c_\lambda\,\sigma\circ\widehat \psi(\gamma)\, x^{\gamma\cdot\lambda} \right) \\ & = \sum_{\lambda\in\wt\Lambda_{P,M}^\sigma}\
|\mathrm{Stab}_{\Gamma_M} \left(\lambda\right)|\,c_\lambda\,x^\lambda+\sum_{\lambda\in\wt\Lambda_{P,M}^\sigma}\sum_{\gamma\in\Gamma_M\setminus
\mathrm{Stab}_{\Gamma_M}\left(\lambda\right)}c_\lambda\,\sigma\circ\widehat \psi(\gamma)\, x^{\gamma\cdot\lambda} . \label{keystab}
\end{align}
Since $MF_P$ is a fundamental domain of $\Gamma_M$ acting on $MF_Q$, the second term of \eqref{keystab} is a linear combination of monomials $x^\la$ with  $\lambda\notin\wt\Lambda^\sigma_{P,M}$.
Since the monomials $x^\la, \,\lambda\in\wt\Lambda^\sigma_{P,M}$ in the first term of \eqref{keystab} are linearly independent, setting \eqref{keystab} equal to zero grants that all  $c_\la \in \C$ with $\lambda\in\wt\Lambda^\sigma_{P,M}$ are also zero. The polynomials \eqref{fla} are thus linearly independent. 
 
For any $(\mathcal{A},\sigma)-$invariant polynomial $f \in\Pi_M^{\mathcal{A},\sigma}$ of the form $$f(x)=\sum_{\la\in\Lambda_M}c_\la \,x^\la$$  
and for any $\gamma\in\Gamma_M$ imply properties \eqref{sigakce}, \eqref{hatpsi}, \eqref{matice} and \eqref{Ax}  that 
\begin{align}\label{invpoly}
\sigma\circ\widehat\psi(\gamma)\sum_{\la\in\Lambda_M}c_\la \,x^\la= \sum_{\la\in\Lambda_M}c_\la \, x^{\gamma\cdot \la}.
\end{align}
The set of weights $\Lambda_M$ is invariant under $\Gamma_M$ and therefore, comparing the coefficients in \eqref{invpoly} yields that $\sigma\circ\widehat\psi(\gamma)c_\lambda=c_{\gamma^{-1}\cdot\lambda}$. Thus, relation $\sigma\circ\widehat\psi(\gamma^{-1})=\sigma\circ\widehat\psi(\gamma)$ guarantees for all $\gamma\in\Gamma_M$ that
 \begin{equation}\label{koef}
 \sigma\circ\widehat\psi(\gamma)c_\lambda= c_{\gamma\cdot\lambda}.
 \end{equation}
The disjoint decomposition \eqref{disLam} of the weight set $\Lambda_M$ together with relations \eqref{gaml1}, \eqref{gaml2} and \eqref{koef} ensures that
\begin{align} \nonumber
f(x)&=\sum_{\gamma\in\Gamma_M} \sum_{\lambda\in\wt\Lambda^\sigma_{P,M}}c_{\gamma\cdot\la} \,x^{\gamma\cdot\la}+\sum_{\gamma\in\Gamma_M}\sum_{\lambda\in \wt{H}^\sigma_{P,M}}c_{\gamma\cdot\la} \,x^{\gamma\cdot\la}\\&= \sum_{\lambda\in\wt\Lambda^\sigma_{P,M}}c_\la f_{\la}^\sigma(x)+\sum_{\lambda\in \wt{H}^\sigma_{P,M}}c_\la\sum_{\gamma\in\Gamma_M}\sigma\circ\widehat\psi(\gamma)\, x^{\gamma\cdot\la}. \label{genn}
\end{align}
Defining relation \eqref{pomlamH} guarantees that for any $\lambda\in \wt{H}^\sigma_{P,M}$ there exists $\wt\gamma\in\mathrm{Stab}_{\Gamma_M}\left(\lambda\right)$ such that $\sigma\circ\widehat\psi(\wt\gamma)=-1$ and thus the second term in  \eqref{genn} vanishes,
$$\sum_{\gamma\in\Gamma_M}\sigma\circ\widehat\psi(\gamma)\,x^{\gamma\cdot\la}=\sum_{\gamma\in\Gamma_M}\sigma\circ\widehat\psi(\gamma)\, x^{\gamma \cdot (\wt\gamma\cdot \la)}=-\sum_{\gamma\in\Gamma_M}\sigma\circ\widehat\psi(\gamma)\, x^{\gamma\cdot\la}=0,$$ 
and the polynomials \eqref{fla} generate the space $\Pi_M^{\mathcal{A},\sigma}$. The constructed  basis $f_{\la}^\sigma, \,\lambda\in\wt\Lambda^\sigma_{P,M}$ provides together with 
Proposition \ref{body} the resulting relation for dimension of  $\Pi_M^{\mathcal{A},\sigma}$,
$$\dim{\Pi^{\mathcal{A},\sigma}_{M}}=\abs{\wt\Lambda^\sigma_{P,M}}=\abs{\Lambda^\sigma_{P,M+m^\sigma}}. $$
\end{proof} 
\begin{tvr}\label{lemD}
The dimension of $\Pi^{\mathcal{D},\sigma}_M$ is for any $M\in\N$ equal to the number of points in $F^\sigma_{Q^\vee,M+m^\sigma}$,
$$\dim{\Pi^{\mathcal{D},\sigma}_M}=\abs{F^\sigma_{Q^\vee,M+m^\sigma}}.$$
\end{tvr}
\begin{proof}
Invariance property \eqref{sigakce} of the diagonal representation \eqref{diagdef} grants that any polynomial from $\Pi_M^{\mathcal{D},\sigma}$ is a linear combination of $(\mathcal{D},\sigma)-$invariant monomials. Therefore, the set of $(\mathcal{D},\sigma)-$invariant monomials $x^\la$, $\la \in \Lambda_M$ satisfying for all $\gamma\in\Gamma$
\begin{equation}\label{Dgamma}
\mathcal{D}(\gamma)\cdot x^\la=\sigma\circ\widehat\psi(\gamma)\,x^\la,	
\end{equation}
forms a basis of $\Pi_M^{\mathcal{D},\sigma}$. Since for any $\gamma_1,\gamma_2\in\Gamma$ it holds that
$$\mathcal{D}(\gamma_1 \gamma_2)\cdot x^\la=\mathcal{D}(\gamma_1) \cdot (\mathcal{D}( \gamma_2)\cdot x^\la)=\sigma\circ\widehat\psi( \gamma_2)\mathcal{D}(\gamma_1) \cdot\,x^\la =\sigma\circ\widehat\psi(\gamma_1\gamma_2)\,x^\la,$$
 verifying property \eqref{Dgamma} only for the generators of $\Gamma$  yields its validity for all $\gamma\in\Gamma$.
For any $\la=[\la_0,\dots,\la_n]\in\Lambda_M$, with its Kac coordinates $\la_i\in\Z^{\geq0}$, $i\in \{0,1,\dots,n\}$ satisfying the defining relation in \eqref{lambda}
\begin{equation}\label{condition1}
\la_0+m_1^\vee\la_1+\dots+m_n^\vee\la_n=M,
\end{equation}
and for any diagonal generator $\mathcal{D}(\gamma)\equiv\mathrm{diag}(d_0,\dots,d_n)$ of $\Gamma$ is condition \eqref{Dgamma} equivalent to the relation
\begin{equation}\label{condition}
d_0^{\la_0}\dots d_n^{\la_n}=\sigma\circ\widehat\psi(\gamma).
\end{equation} 
Explicit forms \eqref{diagonal} of the diagonal generators yield the following explicit reformulations of \eqref{condition} for the non-trivial cases,
\begin{equation}\label{condition2}
\begin{alignedat}{2}
&A_n:&\q & \la_1+2\la_2+3\la_3+\dots+n\la_n\equiv h^\sigma \mod {n+1},\\
&B_{2k+1}:&\q &\la_1+\la_3+\dots+\la_{2k+1} \equiv h^\sigma \mod 2,\\
&B_{2k}:&\q&\la_3+\la_5+\dots+\la_{2k-1}+\la_{2k}\equiv h^\sigma \mod 2,\\
&C_n:&\q &\la_1\equiv h^\sigma \mod 2,\\
&D_{2k+1}:&\q &2\la_1+2\la_3+\dots+2\la_{2k-1}+\la_{2k}+3\la_{2k+1}\equiv h^\sigma \mod 4,\\
&D_{2k}:&\q & \la_{2k-1}+\la_{2k}\equiv 0 \mod 2,\\
&&\q &\la_1+\la_3+\dots+\la_{2k-1}\equiv h^\sigma \mod 2,\\
&E_6:&\q &\la_1+2\la_2+\la_4+2\la_5\equiv 0 \mod 3,\\
&E_7:&\q & \la_4+\la_6+\la_7\equiv h^\sigma\mod 2,
\end{alignedat}
\end{equation}
where non-zero values of $h^\sigma$ are listed in Table~\ref{kongr}.
\begin{table}
\begin{tabular}{|c|c|c|c|c|c|c|c|c|}
\hline
&$A_{2k+1}$&$B_{4k+1}$&$B_{4k+2}$&$B_{4k+3}$&$C_n$&$D_{4k+2}$&$D_{4k+3}$&$E_7$\\
\hline\hline
$h^{\sigma^e}$&$k+1$&$1$&$1$&$0$&$1$&$1$&$2$&$1$\\
$h^{\sigma^s}$&$-$&$1$&$0$&$1$&$0$&$-$&$-$&$-$\\
$h^{\sigma^l}$&$-$&$0$&$1$&$1$&$1$&$-$&$-$&$-$
\\\hline
\end{tabular}
\bigskip
\caption{The non-zero values of symbols $h^\sigma$.}
\label{kongr}
\end{table}

Recall from Proposition 2.1 in \cite{HMPdis} that the short and long Coxeter numbers \eqref{msigma} are of the form
\begin{equation}\label{slcoxnum}
m^{\sigma^s}=\sum_{\alpha_i\in\Delta_s}m_i,\q m^{\sigma^l}=\sum_{\alpha_i\in\Delta_l}m_i+1.
\end{equation} 
Introducing the symbols $s_i\in\Z^{\geq0}$ by relations
\begin{alignat*}{4}
&s_i^\id=s_i&\q&i=0,\dots,n,\\
&s_i^{\sigma^e}=1+s_i&\q&i=0,\dots,n,\\
&s_i^{\sigma^s}=1+s_i,&\q&\alpha_i\in\Delta_s,&\q\q&s_0^{\sigma^s}=s_0,\,s_i^{\sigma^s}=s_i&\q&\alpha_i\in\Delta_l,  \\
&s_i^{\sigma^l}=s_i, &\q&\alpha_i\in\Delta_s,&\q\q& s_0^{\sigma^l}=1+s_0,\,s_i^{\sigma^l}=1+s_i&\q&\alpha_i\in\Delta_l,
\end{alignat*}
and substituting them into the expressions \eqref{fundQvee} and \eqref{cond2} together with formulas \eqref{coxnum} and \eqref{slcoxnum} implies, that the number of points in $F^\sigma_{Q^\vee,M+m^\sigma}$ is equal to the number of solutions of the equations
\begin{equation}\label{condition3}
s_0+m_1 s_1+\dots+ m_n s_n=M
\end{equation}
and
\begin{equation}\label{condition4}
\begin{alignedat}{2}
&A_n:&\q & s_1+2s_2+3s_3+\dots+ ns_n\equiv h^\sigma \mod {n+1},\\
&B_{2k+1}:&\q & s_1+s_3+s_5+\dots+s_{2k+1}\equiv h^\sigma \mod 2 ,\\
&B_{2k}:&\q & s_1+s_3+s_5+\dots+s_{2k-1}\equiv h^\sigma \mod 2 ,\\
&C_n:&\q & s_n\equiv h^\sigma \mod 2,\\
&D_{4k:}&\q & s_1+s_3+\dots+s_{4k-3}+s_{4k-1}\equiv 0 \mod 2,\\
&&& s_{4k-1}+s_{4k}\equiv 0 \mod 2,\\
&D_{4k+1}:&\q & 2s_1+2s_3+\dots+ 2s_{4k-1}+3s_{4k}+s_{4k+1} \equiv 0 \mod 4,\\
&D_{4k+2}:&\q &s_1+s_3+\dots+s_{4k-1}+s_{4k+2}\equiv h^\sigma \mod 2,\\
&&& s_{4k+1}+s_{4k+2}\equiv 0 \mod 2,\\
&D_{4k+3}:&\q &2s_1+2s_3+\dots+ 2s_{4k+1}+s_{4k+2}+3s_{4k+3} \equiv h^\sigma \mod 4,\\
&E_6:&\q & s_1+2s_2+s_4+2s_5\equiv 0 \mod 3,\\
&E_7:&\q & s_4+s_6+s_7\equiv h^\sigma \mod 2.
\end{alignedat}
\end{equation}
The solutions $s_i\in\Z^{\geq0}$ of  the equations \eqref{condition3} and \eqref{condition4} and solutions $\la_i\in\Z^{\geq0}$ of the equations \eqref{condition1} and \eqref{condition2} coincide for each case up to a permutation and therefore, their numbers are identical.
\end{proof}

\begin{thm}\label{pocet}
For any $M\in\N,M>m^\sigma$ it holds that
$$\abs{F^\sigma_{Q^\vee,M}}=\abs{\Lambda^\sigma_{P,M}}.$$
\end{thm}
\begin{proof}
Combining Propositions \ref{lemA}, Proposition \ref{lemD} and Theorem \ref{isoAD} grants for any $M\in\N$ that
$$\abs{F^\sigma_{Q^\vee,M+m^\sigma}}=\dim{\Pi^{\mathcal{D},\sigma}_M}=  \dim{\Pi^{\mathcal{A},\sigma}_M}=\abs{\Lambda^\sigma_{P,M+m^\sigma}}.  $$
\end{proof}

\subsection{Counting formulas}\

In order to obtain explicit counting formulas for cardinalities of all $\Lambda^\sigma_{P,M}$ and $M> m^\si$, the following crucial identity stemming from 
Proposition \ref{body} is used,
\begin{equation}\label{revert}
\abs{\Lambda^\sigma_{P,M}}=\abs{\wt\Lambda^\sigma_{P,M-m^\sigma}}.	
\end{equation}
The calculation of the cardinality of $\Lambda^\sigma_{P,M}$ is thus reverted to counting the weights in $\wt\Lambda^\sigma_{P,M}$ for all $M\in \N$. The group $\Gamma_M$ partitions the sets of weights \eqref{wLQM} and \eqref{wHQM} into $\Gamma_M-$orbits and the sets $\wt\Lambda^\sigma_{P,M}$ and $\wt H^\sigma_{P,M}$  consist of exactly one point from each $\Gamma_M-$orbit. Therefore, the number of points in $\wt \Lambda_{P,M}^\si$ and $\wt H_{P,M}^\si$ is equal to the number of $\Gamma_M-$orbits in  $\wt \Lambda_{Q,M}^\si$ and $\wt H_{Q,M}^\si$, respectively.

Starting with the identity sign homomorphism set $\wt \Lambda^\id_{P,M}$ and introducing a set of points in $\Lambda_M$ fixed  by a given $\gamma\in \Gamma_M$,
\begin{equation*}
\mathrm{Fix}_M(\gamma)\equiv\setb{\lambda\in\Lambda_M}{\gamma\cdot\lambda=\lambda},	
\end{equation*}
the Burnside's lemma applied to the set $\wt \Lambda_{Q,M}^\id= \Lambda_M$ provides the relation
\begin{equation}\label{burn}
\abs{ \Lambda_{P,M}^\id}=\abs{\wt \Lambda_{P,M}^\id}=\frac{1}{c}\sum_{\gamma\in\Gamma_M}\abs{\mathrm{Fix}_M(\gamma)},\q
\end{equation}
Moreover, the group $\Gamma_M$ is partitioned into disjoint sets containing the elements of the same order $|\gamma|=d$ which implies that \eqref{burn} is reformulated as
\begin{equation}\label{burn2}
\abs{\Lambda_{P,M}^\id}=\frac{1}{c}\sum_{d\mid c}\sum_{\substack{\abs{\gamma}=d\\\gamma\in\Gamma_M}}\abs{\mathrm{Fix}_M(\gamma)}.
\end{equation}

For a general sign homomorphism, disjoint decomposition \eqref{disPLam} yields the relation
\begin{equation}\label{countid}
\abs{\wt\Lambda^\sigma_{P,M}}=\abs{\Lambda_{P,M}^\id}-\abs{\wt H_{P,M}^\si},
\end{equation}
and employing the Burnside's lemma to the set $\wt H_{Q,M}^\si$ produces the identity
\begin{equation}\label{burn3}
\abs{\wt H_{P,M}^\si}=\frac{1}{c}\sum_{d\mid c}\sum_{\substack{\abs{\gamma}=d\\\gamma\in\Gamma_M}}\abs{\mathrm{Fix}_M(\gamma)\cap \wt H_{Q,M}^\si}.
\end{equation}
 Using the Euler's totient function $\phi$ together with equations \eqref{revert}, \eqref{burn2}, \eqref{countid} and \eqref{burn3}, the explicit counting formulas for the numbers of points in the weight sets $\Lambda^\sigma_{P,M}$ are for all cases listed in the following theorem.
\begin{thm}
The numbers of elements in $\Lambda^\sigma_{P,M}$ are for any $M\in\N, M>m^\sigma$ determined by the following formulas.
\begin{enumerate}
\item $A_n (n\geq1)$:
{\small 
\begin{align}
&\abs{\Lambda^\id_{P,M}(A_n)}=\frac{1}{n+M+1}\sum_{d\,| \gcd{(n+1,M)}} \varphi(d) \begin{pmatrix}\frac{n+M+1}{d}\\ \frac{n+1}{d} \end{pmatrix}, \label{pocetAn}\\
&\nonumber\\
&\abs{\Lambda^{\sigma^e}_{P,M}(A_n)}=\frac{(-1)^{n+1}}{M}\sum_{d\,| \gcd{(n+1,M)}} (-1)^{\tfrac{n+1}{d}}\varphi(d) \begin{pmatrix}\frac{M}{d}\\ \frac{n+1}{d} \end{pmatrix},\label{SpocetAn}
\end{align}}
\item $B_n (n\geq 3)$:
{\small
\begin{align*}
&\abs{\Lambda^\id_{P,M}(B_n)}=\begin{cases}
{{n+k}\choose n}&M=2k+1,\\
\frac12\left[{{2m+2k+l+1}\choose {2m+1}}+{{2m+2k+l}\choose {2m+1}}+{{m+k\choose m}}\right]&M=4k+2l,n=2m+1,\\
\frac12\left[{{2m+2k+l}\choose {2m}}+{{2m+2k+l-1}\choose {2m}}+{{m+k\choose m}}+{{m+k+l-1}\choose m}\right]&M=4k+2l,n=2m,
\end{cases}\\
&\\
&\abs{\Lambda^{\sigma^e}_{P,M}(B_n)}=\begin{cases}
\frac12\left[{{2k+l}\choose {4m+1}}+{{2k+l-1}\choose {4m+1}}-{{k+l-1\choose 2m}}\right]&M=4k+2l,n=4m+1,\\
\frac12\left[{{2k+l}\choose {4m+2}}+{{2k+l-1}\choose {4m+2}}-{{k\choose 2m+1}}-{{k+l-1}\choose 2m+1}\right]&M=4k+2l,n=4m+2,\\
\abs{\Lambda^\id_{P,M-2n}(B_n)}&\text{otherwise},\\
\end{cases}\\
&\\
&\abs{\Lambda^{\sigma^s}_{P,M}(B_n)}=\begin{cases}
\frac12\left[{{2m+2k+l}\choose {2m+1}}+{{2m+2k+l-1}\choose {2m+1}}-{{m+k+l-1\choose m}}\right]&M=4k+2l,n=2m+1,\\
\abs{\Lambda^\id_{P,M-2}(B_n)}&\text{otherwise},
\end{cases}\\
&\\
&\abs{\Lambda^{\sigma^l}_{P,M}(B_n)}=\begin{cases}
\frac12\left[{{2k+l+1}\choose {4m+2}}+{{2k+l}\choose {4m+2}}-{{k\choose 2m+1}}-{{k+l}\choose 2m+1}\right]&M=4k+2l,n=4m+2,\\
\frac12\left[{{2k+l+1}\choose {4m+3}}+{{2k+l}\choose {4m+3}}-{{k\choose 2m+1}}\right]&M=4k+2l,n=4m+3,\\
\abs{\Lambda^\id_{P,M-2n+2}(B_n)}&\text{otherwise},
\end{cases}
\end{align*}}
where $l\in \{0,1\}$.\\
\item $C_n (n\geq 2)$:
{\small
\begin{align*}
&\abs{\Lambda^\id_{P,M}(C_n)}={{n+k}\choose{n}},\q \text{where } k=\left\lfloor \tfrac{M}{2}\right\rfloor,\\
&\abs{\Lambda^{\sigma^e}_{P,M}(C_n)}=\abs{\Lambda^\id_{P,M-2n-1}(C_n)}, \\
&\abs{\Lambda^{\sigma^s}_{P,M}(C_n)}=\abs{\Lambda_{P,M-2n+2}^\id(C_n)},\\
&\abs{\Lambda^{\sigma^l}_{P,M}(C_n)}=\abs{\Lambda_{P,M-3}^\id(C_n)}.
\end{align*}}
\\
\item $D_n (n\geq4)$:

{\small
\begin{align*}
&\abs{\Lambda_{P,M}^\id(D_n)}=\begin{cases}
{{n+k}\choose {n}}+{{n+k-1}\choose {n}}&M=2k+1,\\
\frac14\left[{{2m+2k+1}\choose{2m+1}}+6{{2m+2k}\choose{2m+1}}+{{2m-1+2k}\choose {2m+1}}+{{2m-1+2k}\choose {2m-1}}+2{{m+k-1}\choose{m-1}}\right]&M=4k,n=2m+1,\\
\frac14\left[{{2m+2k+2}\choose{2m+1}}+6{{2m+2k+1}\choose{2m+1}}+{{2m+2k}\choose {2m+1}}+{{2m+2k}\choose {2m-1}}\right]&M=4k+2,n=2m+1,\\
\text{{\tiny$\begin{aligned}
\frac14&\left[{{2m+2k}\choose{2m}}+6{{2m+2k-1}\choose{2m}}+{{2m+2k-2}\choose{2m}}+{{2m+2k-2}\choose{2m-2}}
\right.\\
&\left.+2{{m+k}\choose{m}}+6{{m+k-1}\choose{m}}\right]
\end{aligned}$}}&M=4k,n=2m,\\
\text{{\tiny$\begin{aligned}\frac14&\left[{{2m+2k+1}\choose{2m}}+6{{2m+2k}\choose{2m}}+{{2m+2k-1}\choose{2m}}+{{2m+2k-1}\choose{2m-2}}
\right. \\
&\left.+6{{m+k}\choose{m}}+2{{m+k-1}\choose{m}}\right]
\end{aligned}$}}&M=4k+2,n=2m,
\end{cases}\\
&\\
&\abs{\Lambda^{\sigma^e}_{P,M}(D_n)}=\begin{cases}
\frac14\left[{{2k+1}\choose{4m+3}}+6{{2k}\choose{4m+3}}+{{2k-1}\choose {4m+3}}+{{2k-1}\choose {4m+1}}-2{{k-1}\choose{2m}}\right]&M=4k,n=4m+3,\\
\frac14\left[{{2k+2}\choose{4m+2}}+6{{2k+1}\choose{4m+2}}+{{2k}\choose{4m+2}}+{{2k}\choose{4m}}
-2{{k+1}\choose{2m+1}}-6{{k}\choose{2m+1}}\right]
&M=4k+2,n=4m+2,\\
\frac14\left[{{2k+1}\choose{4m+2}}+6{{2k}\choose{4m+2}}+{{2k-1}\choose{4m+2}}+{{2k-1}\choose{4m}}
-6{{k}\choose{2m+1}}-2{{k-1}\choose{2m+1}}\right]
&M=4k,n=4m+2,\\
\abs{\Lambda^\id_{P,M-2n+2}(D_n)}&\text{otherwise}.
\end{cases}
\end{align*}
}\\
\item $E_6$:
{\small
\begin{align*}
&\abs{\Lambda^\id_{P,M}(E_6)}=\begin{cases}
\frac13\left[ \abs{\Lambda_{6k}(E_6)}+2{{k+2}\choose{2}}+2{{k+1}\choose{2}}\right]&M=6k,\\
\frac13\left[ \abs{\Lambda_{6k+3}(E_6)}+4{{k+2}\choose{2}}\right]&M=6k+3,\\
\frac13\abs{\Lambda_M(E_6)}&\text{otherwise},
\end{cases}\\
&\\
&\abs{\Lambda^{\sigma^e}_{P,M}(E_6)}=\abs{\Lambda^\id_{P,M-12}(E_6)}.
\end{align*}
}
\item $E_7$:
{\small
\begin{align*}
&\abs{\Lambda^\id_{P,M}(E_7)}=\begin{cases}
\frac12 \abs{\Lambda_{2k+1}(E_7)}&M=2k+1,\\
\frac12\left[\abs{\Lambda_{12k+2l}(E_7)}+\sum_{i=0}^3d_{li}{{4-i+k}\choose{4}}\right] &M=12k+2l,l\in\{0,\dots,5\},
\end{cases}\\
&\\
&\abs{\Lambda^{\sigma^e}_{P,M}(E_7)}=\begin{cases}
\frac12\left[\abs{\Lambda_{12k+2l-18}(E_7)}-\sum_{i=0}^3d_{l+3,i}{{2-i+k}\choose{4}}\right]&M=12k+2l,l\in\{0,1,2\},\\
\frac12\left[\abs{\Lambda_{12k+2l-18}(E_7)}-\sum_{i=0}^3d_{l-3,i}{{3-i+k}\choose{4}}\right]&M=12k+2l,l\in\{3,4,5\},\\
\abs{\Lambda^\id_{P,M-18}(E_7)}&\text{otherwise},\\
\end{cases}
\end{align*}
}
where 
{\small $$\left(d_{li}\right)=\left(\begin{matrix}
1&34&64&9\\2&46&55&5\\5&55&46&2\\9&64&34&1\\16&67&25&0\\25&67&16&0
\end{matrix}\right).$$}\\
\item $E_8,F_4,G_2$: {\small $$\abs{\Lambda_{P,M}^\sigma}=\abs{\Lambda^\sigma_{Q,M}}.$$}
\end{enumerate}
\end{thm}
\begin{proof}
A detailed calculation is presented for the infinite series of groups $\Gamma_M$ of $A_n$. Since the remaining infinite series $B_n$, $C_n$ and $D_n$ share common groups $\Gamma_M$ listed in Table \ref{gamtab}, the proof of the corresponding counting formulas is less complex. 

The group $\Gamma_M$ of $A_n$ is a cyclic group of order $n+1$ generated by the permutation $\gamma_1$ satisfying $\gamma_1^{n+1}=1$,  $$\Gamma_M=\{\gamma_1,\gamma_1^2,\dots,\gamma_1^{n+1}\}.$$ 
For each $k\in\{1,\dots,n+1\}$ such that order $\abs{\gamma_1^k}=d$ it holds that
\begin{equation*}
d=\frac{n+1}{l},\q l=\mathrm{gcd}(n+1,k).
\end{equation*}
Since for any weight $\la\in\mathrm{Fix}_M\left(\gamma_1^k\right)$, its cyclic isotropy subgroup $\mathrm{Stab}_{\Gamma_M}(\la)$ contains both $\gamma_1^k$ and $\gamma_1^l$, it holds that $\la\in\mathrm{Fix}_M\left(\gamma_1^l\right)$. Conversely, as $l$ divides $k$, the set of fixed points $\mathrm{Fix}_M\left(\gamma_1^l\right)$ is contained in the set $\mathrm{Fix}_M\left(\gamma_1^k\right)$ and thus
\begin{equation}\label{rovnost}
\mathrm{Fix}_M\left(\gamma_1^k\right)=\mathrm{Fix}_M\left(\gamma_1^l\right).
\end{equation} 
Since the number of elements of order $d$ in a cyclic group $\Gamma_M$ is the value of the Euler's totient function $\phi(d)$ and all elements of order $d$ satisfy \eqref{rovnost}, 
counting relation \eqref{burn2} specializes to
\begin{equation}\label{countAn}
\abs{\Lambda_{P,M}^\id}=\frac{1}{n+1}\sum_{d\mid n+1} \phi(d)  \abs{\mathrm{Fix}_M(\gamma_1^{l})}.
\end{equation}
The Kac coordinates $\la_0,\dots,\la_n\in\Z^{\geq0}$ of a point $\lambda\in\Lambda_{M}$ satisfy equation \eqref{condition1} specialized to $A_n$,
\begin{equation}\label{eq1}
\la_0+\la_1+\dots+\la_n=M.
\end{equation}
Employing explicit expression for the cyclic permutation listed in Table~\ref{gamtab}, a weight $\la\in \Lambda_M$ is stabilized by $\gamma_1^l\in \Gamma_M$, i.e. $\lambda\in \mathrm{Fix}_{M}\left(\gamma_1^l\right)$, if and only if 
 \begin{equation}\label{eq3}
 \la_i=\la_{i+l}=\dots=\la_{i+(d-1)l},\q i=0,\dots,l-1.
 \end{equation}
Substituting \eqref{eq3} into relation \eqref{eq1} yields 
\begin{equation}\label{eq2}
d(\la_0+\dots+\la_{l-1})=M,
\end{equation}
and thus, if $M$ is not divisible by $d$, then 
\begin{equation}\label{ndivM}
\abs{\mathrm{Fix}_M\left(\gamma_1^l\right)}=0, \q d\nmid M.	
\end{equation}
If $d$ divides $M$, then the number $\abs{\mathrm{Fix}_M(\gamma_1^l)}$ coincides with the number of non-negative integer solutions of equation \eqref{eq2}. 
This number is determined by Proposition 3.1 in \cite{HP} as
\begin{equation}\label{Fix}
\abs{\mathrm{Fix}_M\left(\gamma_1^l\right)}={{\tfrac{n+1+M}{d}-1}\choose{\tfrac{n+1}{d}-1}}=\frac{n+1}{n+M+1}{{\tfrac{n+1+M}{d}}\choose{\tfrac{n+1}{d}}}, \q d\mid M.	
\end{equation}
Substituting  relations \eqref{Fix} and \eqref{ndivM} into \eqref{countAn} results in counting formula \eqref{pocetAn}. 

Taking into account \eqref{hatpsi}, the values of the composition of homomorphisms $\sigma^e\circ\widehat\psi$ on the elements of the group $\Gamma_M$ follow for the case $A_n$ from Table \ref{gamtab} as
\begin{equation}\label{signAn}
\sigma^e\circ\widehat\psi(\gamma^k_1)=(-1)^{nk}.	
\end{equation}
For $n$ even, definition \eqref{wHQM} and relation \eqref{signAn} immediately guarantee that $\wt H_{Q,M}^{\si^e}=\emptyset$ and thus
\begin{equation}\label{zeroeven}
\abs{\wt H_{P,M}^{\si^e}}=0,\q n \,\, \mathrm{even.}	
\end{equation}
Continuing with $n$ odd, equality \eqref{rovnost} again grants that counting formula \eqref{burn3} simplifies as
\begin{equation}\label{countHAn}
\abs{\wt H_{P,M}^{\si^e}}=\frac{1}{n+1}\sum_{d\mid n+1} \phi(d)  \abs{\mathrm{Fix}_M(\gamma_1^{l}) \cap\wt H_{Q,M}^{\si^e} }=\frac{1}{n+1}\sum_{l\mid n+1} \phi\left(\frac{n+1}{l}\right)  \abs{\mathrm{Fix}_M(\gamma_1^{l}) \cap\wt H_{Q,M}^{\si^e} }.
\end{equation}
In order to determine the structure of the intersection sets in \eqref{countHAn}, note that for an odd $n$ there exist an odd number $m$ and $j\in \N$ such that 
$$n+1=2^jm,$$
and any divisor $l$ of $n+1$ is of the form
\begin{equation}\label{divl}
l= 2^i p, \q p \mid m,\,  i\in \{0,1,\dots,j\}. 	
\end{equation}
Taking any $\la\in \mathrm{Fix}_M(\gamma_1^{l}) \cap\wt H_{Q,M}^{\si^e}$, relations \eqref{signAn} and \eqref{wHQM} yield that there exists an odd number $u\in \{1,\dots,n+1\}$ such that 
$\gamma_1^u\in\mathrm{Stab}_{\Gamma_M}(\la)$. The stabilizer subgroup $\mathrm{Stab}_{\Gamma_M}(\la)$ therefore contains a subgroup generated by $\gamma_1^u, \gamma_1^l\in \mathrm{Stab}_{\Gamma_M}(\la),$
$$\setb{(\gamma_1^u)^r (\gamma_1^l)^t}{r,t\in \Z}=\setb{\gamma_1^{ur +lt}}{r,t\in \Z} \subset \mathrm{Stab}_{\Gamma_M}(\la) ,  $$
and thus $\gamma_1^{ \gcd{(u,l)}}\in \mathrm{Stab}_{\Gamma_M}(\la).$ Since $u$ is odd, the greatest common divisor $\gcd{(u,l)}$ is odd and divides $p$ and therefore 
$\la \in  \mathrm{Fix}_M(\gamma_1^{ \gcd{(u,l)}})\subset  \mathrm{Fix}_M(\gamma_1^{p})$. Conversely, taking any $\la \in\mathrm{Fix}_M(\gamma_1^{p})\subset \mathrm{Fix}_M(\gamma_1^{l})$, the divisor $p$, being odd, grants that $\la \in\wt H_{Q,M}^{\si^e}$ and hence
\begin{equation}\label{inter}
\mathrm{Fix}_M(\gamma_1^{l}) \cap\wt H_{Q,M}^{\si^e}=\mathrm{Fix}_M(\gamma_1^{p}).
\end{equation}
Relations \eqref{inter} and \eqref{divl} allow to further evaluate the counting formula \eqref{countHAn} as
\begin{equation}\label{countHAn2}
\abs{\wt H_{P,M}^{\si^e}}=\frac{1}{n+1}\sum_{l\mid n+1} \phi\left(\frac{n+1}{l}\right)  \abs{\mathrm{Fix}_M(\gamma_1^{p})  }= \frac{1}{n+1}\sum_{p\mid m}\sum_{i=0}^j \phi\left(\frac{n+1}{2^i p}\right)  \abs{\mathrm{Fix}_M(\gamma_1^{p})  }.
\end{equation}
Since the Euler's totient function $\phi$ is multiplicative and $\phi(2^i)=2^{i-1}$ holds for any $i\in \N$, the following identities are obtained,
$$\sum_{i=0}^j \phi\left(\frac{n+1}{2^i p}\right)=\sum_{i=0}^j \phi\left(2^i\frac{m}{ p}\right)=\phi \left(\frac{m}{ p}\right)\sum_{i=0}^j \phi\left(2^i\right) =2^j \phi \left(\frac{m}{ p}\right)=2\phi \left(\frac{n+1}{ p}\right), $$
and thus \eqref{countHAn2} is simplified as
\begin{align}
\abs{\wt H_{P,M}^{\si^e}}&= \frac{2}{n+1}\sum_{p\mid m} \phi\left(\frac{n+1}{ p}\right)  \abs{\mathrm{Fix}_M(\gamma_1^{p})  }=\frac{1}{n+1}\sum_{l\mid n+1}\left(1-(-1)^l\right) \phi\left(\frac{n+1}{ l}\right)  \abs{\mathrm{Fix}_M(\gamma_1^{l})  }\nonumber   \\
 &=\frac{1}{n+1}\sum_{d\mid n+1}\left(1-(-1)^{l}\right) \phi\left(d\right)  \abs{\mathrm{Fix}_M(\gamma_1^{l})  }. \label{countHAn3}
\end{align}
The resulting counting formula 
\begin{equation}\label{countres}
\abs{\wt H_{P,M}^{\si^e}}=	\frac{1}{n+1}\sum_{d\mid n+1}\left(1-(-1)^{l+n+1}\right) \phi\left(d\right)  \abs{\mathrm{Fix}_M(\gamma_1^{l})  }
\end{equation}
specializes for $n$ odd to \eqref{countHAn3} and for $n$ even to \eqref{zeroeven}.  Substituting \eqref{ndivM} and \eqref{Fix} into the counting formula \eqref{countres} and the result into \eqref{countid} and  \eqref{revert}, while taking into account that the Coxeter number of $A_n$ is $m^{\si^e}=n+1$, yields the final counting relation \eqref{SpocetAn}.
\end{proof}

\section{Discrete transforms of Weyl orbit functions}
\subsection{Weyl and Hartley orbit functions}\

The sign homomorphisms $\sigma$ of the Weyl group $W$ induce up to four families of Weyl orbit functions. The Weyl orbit functions $\varphi^\sigma_b :\R^n \map \C$, labelled by parameter $b\in \R^n$, are defined for any $a\in \R^n$ by 
\begin{equation}\label{Weylorb}
\varphi^\sigma_b(a)\equiv\sum_{w\in W}\sigma(w)\,e^{2\pi \i \langle wb,a\rangle}.
\end{equation}
The functions $\varphi^\id_b$ and $\varphi^{\si^e}_b$ are called $C-$functions and $S-$functions, respectively. In the cases of simple Lie algebras with two different root lengths, the functions $\varphi^{\si^s}_b$ and $\varphi^{\si^l}_b$ are termed the $S^s-$ and $S^l-$functions in \cite{MMP}, respectively. 

Recall from Proposition 3.1 in \cite{CzHr} that while restricting the label to the the weight lattice $\la\in P$, the argument invariance of Weyl orbit functions with respect to the action of the element of the affine Weyl group $w^{\mathrm{aff}}\in W^{\mathrm{aff}}$ is for any $a\in \R^n$ of the form
\begin{equation}\label{Wcond3}
\varphi^\sigma_\la(w^{\mathrm{aff}}a)=\sigma\circ\psi (w^{\mathrm{aff}}) \cdot \, \varphi^\sigma_\la(a).
\end{equation}
Moreover, the functions $\varphi^\sigma_\la$ are zero for the boundary points $a'\in F \setminus F^\sigma $ of the sets \eqref{domainF2},
\begin{equation}\label{Wcond3z}
\varphi^\sigma_\la(a')=0,\q  a'\in F \setminus F^\sigma.
\end{equation}
Restricting the arguments of Weyl orbit functions to the refined dual root lattice $\tfrac{1}{M}Q^\vee$, additional label symmetries with respect to the extended dual affine Weyl group \eqref{WaffP} are generated.
\begin{tvr}\label{labelthm}
Let $s \in \frac{1}{M}Q^{\vee},$ then for any $ w_P\in W_P^{\mathrm{aff}}$ and any $b \in \R^n$ it holds that
\begin{equation}\label{invphi}
\varphi_{M w_P \left(\frac{b}{M}\right)}^\si(s)= \si\circ\widehat{\psi}(w_P)\cdot\varphi_{b}^\si (s).	
\end{equation}
Additionally,  the Weyl orbit function $\varphi_{b'}^\si$ vanishes for any $b'\in M(F_P\setminus F_P^\sigma)$,
\begin{equation}\label{vanphi}
\phi_{b'}^{\si}(s)=0, \q b'\in M(F_P\setminus F_P^\sigma).	
\end{equation}
\end{tvr}
\begin{proof}
The $\Z-$duality \eqref{P} of the Weyl group invariant lattices $P$ and $Q^{\vee}$ ensure that the relation$$\langle Mw p,s\rangle=\langle w p,Ms\rangle \in \Z$$ is valid for any $p\in P$ and $w\in W$ and thus, for any element of the extended dual affine Weyl group $w_P=T(p)w \in W^{\mathrm{aff}}_P$ it holds that
$$\varphi_{M w_P \left(\frac{b}{M}\right)}^\si(s)= \sum_{w'\in W}\sigma(w')\,e^{2\pi \i \langle w' w b+ Mw' p,s\rangle}=  \sum_{w'\in W}\sigma(w')\,e^{2\pi \i \langle w' w b,s\rangle}=\sigma(w) \varphi_{b}^\si (s).$$
For points from the boundary set $b'\in M(F_P\setminus F_P^\sigma)$ there exists by \eqref{FsiP} an element $w_P= \mathrm{Stab}_{W^{\mathrm{aff}}_P} (b'/M)$ such that $\sigma\circ\widehat\psi(w_P)=-1$ and hence
\begin{equation*}
	\phi_{b'}^{\si}(s)=\varphi_{M w_P \left(\frac{b'}{M}\right)}^\si(s)= \si\circ\widehat{\psi}(w_P)\cdot\varphi_{b'}^\si (s)=-\phi_{b'}^{\si}(s).
\end{equation*}
\end{proof}

A real-valued version of the Weyl orbit functions employs the Hartley kernel functions
\begin{equation}\label{cas}
\mathrm{cas}\,\alpha=\cos{\alpha}+\sin{\alpha}, \q \alpha\in \R.
\end{equation}
The sign homomorphisms $\sigma$ of the Weyl group $W$ induce up to four families of the Hartley orbit functions $\hart^\sigma_b :\R^n \map \R$, labelled by parameter $b\in \R^n$, 
\begin{equation*}
\hart^\sigma_b(a)=\sum_{w\in W}\sigma(w)\mathrm{cas}\,(2\pi\langle wb,a\rangle).
\end{equation*}
The relation between the complex exponential function in \eqref{Weylorb} and the Hartley kernel function \eqref{cas} provides the following expression,
\begin{equation}\label{ReIm}
	\hart^\sigma_b = \mathrm{Re}\, \phi_{b}^{\si}+\mathrm{Im}\, \phi_{b}^{\si}.
\end{equation}
Equation \eqref{ReIm} straightforwardly extracts from \eqref{Wcond3} that restricting the label to the weight lattice $\la\in P$, the argument invariance of Hartley orbit functions with respect to the action of the element of the affine Weyl group $w^{\mathrm{aff}}\in W^{\mathrm{aff}}$ is for any $a\in \R^n$ of the form
\begin{equation}\label{Hcond}
\hart^\sigma_\la(w^{\mathrm{aff}}a)=\sigma\circ\psi (w^{\mathrm{aff}}) \cdot \, \hart^\sigma_\la(a).
\end{equation}
Moreover, the functions $\hart^\sigma_\la$ are zero for the boundary points $a'\in F \setminus F^\sigma $ of the sets \eqref{domainF2},
\begin{equation}\label{Hcondz}
\hart^\sigma_\la(a')=0,\q  a'\in F \setminus F^\sigma.
\end{equation}
Using \eqref{ReIm} again, a real-valued analogue of Proposition \ref{labelthm} is derived.
\begin{tvr}\label{labelthmH}
Let $s \in \frac{1}{M}Q^{\vee},$ then for any $ w_P\in W_P^{\mathrm{aff}}$ and any $b \in \R^n$ it holds that
$$\hart_{M w_P \left(\frac{b}{M}\right)}^\si(s)= \si\circ\widehat{\psi}(w_P)\cdot\hart_{b}^\si (s).$$
Additionally,  the Hartley orbit function $\hart_{b'}^\si$ vanishes for any $b'\in M(F_P\setminus F_P^\sigma)$,
$$\hart_{b'}^{\si}(s)=0, \q b'\in M(F_P\setminus F_P^\sigma). $$
\end{tvr}

\subsection{Discrete orthogonality}\

Argument  symmetries of Weyl orbit functions \eqref{Wcond3}, \eqref{Wcond3z} grant that the functions $\varphi^\si_b, \la\in P$ discretized to the grid  $\tfrac{1}{M}Q^\vee$ are fully determined by their non-zero values in the sets  \eqref{FQvee}. Label symmetries, induced by the selection of the discrete grid  $\tfrac{1}{M}Q^\vee$ and determined by Proposition \ref{labelthm}, restrict the labels to the weight sets \eqref{LP}. The same argument and label symmetries are guaranteed for the Hartley orbit functions by  relations \eqref{Hcond}, \eqref{Hcondz} and Proposition \ref{labelthmH}. A scalar product of two discrete complex valued functions $f,g:F_{Q^\vee,M}^\sigma\map \C$ is introduced via relation
\begin{equation}\label{scal}
\langle f,g\rangle_M^\sigma\equiv\sum_{s\in F_{Q^\vee,M}^\sigma} \varepsilon(s) f(s)\overline{g(s)},
\end{equation}
with the function $\varepsilon$ defined by \eqref{ep}, and the resulting finite-dimensional Hilbert space is denoted by $\mathcal{H}_{Q^\vee,M}^\si$.
Note that for any $s\in \tfrac{1}{M}Q^\vee$ and $\la \in P$ it holds
$$e^{2\pi\i\langle \la,s+Q^\vee\rangle}=e^{2\pi\i\langle \la,s\rangle},$$
and the mapping $e^{2\pi\i\langle\la,y\rangle}$ with $y\in\tfrac{1}{M}Q^\vee/Q^\vee$ is well defined.
In order to derive the discrete orthogonality of Weyl orbit functions \eqref{Weylorb} with respect to the scalar product \eqref{scal}, the following basic orthogonality relations of exponential functions are essential.
\begin{tvr}\label{lemdis}
Let $M\in\N$ and $\la,\la'\in P$, then
\begin{equation}\label{dddeee}
\sum_{y\in \frac{1}{M}Q^\vee/ Q^\vee}e^{2\pi\i\langle\la-\la',y\rangle}=
\begin{cases}
M^n&\la-\la'\in MP,\\
0&\text{otherwise}.
\end{cases}
\end{equation}
\end{tvr}
\begin{proof}
For any $\la\in P$, such that for all $y\in\tfrac{1}{M}Q^\vee/Q^\vee$ is satisfied $\langle\la,y\rangle\in\Z$, the $\Z-$duality of  $P$ and $Q^\vee$ yields that $\la\in MP$. Thus
for all $\la\in P\setminus MP$, there exists $y'\in \tfrac{1}{M}Q^\vee/Q^\vee$ such that $\langle\la,y'\rangle \notin \Z$ and $e^{ 2\pi \i\langle\la,y'\rangle}\neq1$. Therefore, the following calculation from \cite{MP2}, specialized for the finite quotient group $\tfrac{1}{M}Q^\vee/Q^\vee$ of order $M^n$,  
 $$e^{ 2\pi \i\langle\la,y'\rangle}\sum_{y\in \frac{1}{M}Q^\vee/ Q^\vee}e^{2\pi\i\langle\la ,y\rangle}= \sum_{y\in \frac{1}{M}Q^\vee/ Q^\vee}e^{2\pi\i\langle\la ,y+y'\rangle}=\sum_{y\in \frac{1}{M}Q^\vee/ Q^\vee}e^{2\pi\i\langle\la ,y\rangle},$$ 
forces basic orthogonality relations \eqref{dddeee}.
\end{proof}
\begin{thm}\label{disortogthm}
For any $\la,\la'\in\Lambda^\sigma_{P,M}$, it holds that
\begin{equation}\label{disortog}
\langle \varphi^\sigma_\la,\varphi^\sigma_{\la'}\rangle_M^\sigma=\abs{W}M^nh_{P,M}(\la)\,\delta_{\la\la'},
\end{equation}
where the coefficients $h_{P,M}$ are defined by \eqref{hPM}.
\end{thm}
\begin{proof}
Since the functions $\phi^{\sigma}_\lambda$ vanish by \eqref{Wcond3z}  on the sets $F\setminus F^\sigma$, the scalar product is evaluated as 
\begin{equation}\label{exprs}
\langle \varphi^\sigma_\la,\varphi^\sigma_{\la'}\rangle_M^\sigma= \sum_{s\in F_{Q^\vee,M}^\sigma}\ep(s)\phi^\sigma_\la(s)\overline{\phi^\sigma_{\la'}(s)}=\sum_{s\in \frac{1}{M}Q^\vee \cap F}\ep(s)\phi^\sigma_\la(s)\overline{\phi^\sigma_{\la'}(s)}.
\end{equation}
The $W^{\mathrm{aff}}-$invariance properties \eqref{epinv} and \eqref{Wcond3} grant that the summands in \eqref{exprs} are $W^{\mathrm{aff}}-$invariant, i.e. for all $w^{\mathrm{aff}}\in W^{\mathrm{aff}}$ it holds that 
\begin{equation}\label{Waffin}
\ep(s)\phi^\sigma_\la(s)\overline{\phi^\sigma_{\la'}(s)}=\ep(w^{\mathrm{aff}}s)
\phi^\sigma_\la(w^{\mathrm{aff}}s)\overline{\phi^\sigma_{\la'}(w^{\mathrm{aff}}s)}.
\end{equation}
The $Q^\vee-$shift invariance in \eqref{Waffin}  and relation \eqref{ept} imply
$$\sum_{s\in \frac{1}{M}Q^\vee \cap F}\ep(s)\phi^\sigma_\la(s)\overline{\phi^\sigma_{\la'}(s)}=\sum_{x\in \left[ \frac{1}{M}Q^\vee/Q^\vee\right] \cap F}\wt\ep(x)\phi^\sigma_\la(x)\overline{\phi^\sigma_{\la'}(x)} $$
and the $W-$invariance in \eqref{Waffin} and relations \eqref{rfun1}, \eqref{rfun2} produce the identity
 $$ \sum_{x\in \left[ \frac{1}{M}Q^\vee/Q^\vee\right] \cap F}\wt\ep(x)\phi^\sigma_\la(x)\overline{\phi^\sigma_{\la'}(x)} =\sum_{y\in\frac{1}{M}Q^\vee/Q^\vee }\phi^\sigma_\la(y)\overline{\phi^\sigma_{\la'}(y)}.  $$
The $W-$invariance of the quotient group $\frac{1}{M}Q^\vee/Q^{\vee}$ is used to further simplify the result,
\begin{align}
\sum_{y\in\frac{1}{M}Q^\vee/Q^\vee }\phi^\sigma_\la(y)\overline{\phi^\sigma_{\la'}(y)} = &\sum_{w'\in W}\sum_{w\in W} \sum_{y\in \frac{1}{M}Q^\vee/Q^{\vee}}\sigma(ww')e^{2\pi\i\sca{w\la-w'\la'}{y}}\nonumber \\ = &\abs{W}\sum_{w'\in W}\sigma(w')\sum_{y \in \frac{1}{M}Q^\vee/Q^{\vee}}e^{2\pi\i\sca{\la-w'\la'}{y}}. \label{ddd}
\end{align}
If $\la-w'\la'\in MP$, then $\la/M=w'\la'/M+p$ holds for some $p\in P$ and $w'\in W$ and thus $\la /M$ and $\la' /M$ are in the same $W^{\mathrm{aff}}_P-$orbit.
Definition \eqref{LP} of the weight set $ \Lambda^\sigma_{P,M}$ guarantees that both $\la /M$ and $\la' /M$ are in the fundamental domain $F_P$ and therefore $\la = \la'$. Contrapositive implication yields that if $\la \neq \la'$, then for all $w'\in W$ it holds that $\la-w'\la'\notin MP$ and hence basic orthogonality relations \eqref{dddeee} grant zero value $\sca{\phi^\sigma_\la}{\phi^\sigma_{\la'}}_M^\si=0$.

If $\la=\la'$, then basic orthogonality relations \eqref{dddeee} guarantee that summands in \eqref{ddd} do not vanish only if $\la-w'\la\in MP$, or equivalently $w'\in \widehat\psi \left(\mathrm{Stab}_{W_P^{\mathrm{aff}}}\left(\la/M\right)\right)$ and thus
\begin{align*}
\abs{W}\sum_{w'\in W}\sigma(w')\sum_{y \in \frac{1}{M}Q^{\vee}/Q^{\vee}}e^{2\pi\i\sca{\la-w'\la}{y}} = & |W|\, M^n\sum_{w'\in \widehat\psi\left(\mathrm{Stab}_{W_P^{\mathrm{aff}}}\left(\frac{\la}{M}\right)\right)}\sigma(w') . 
\end{align*}
Since for $w_P\in \mathrm{Stab}_{W_P^{\mathrm{aff}}}(\la/M)$ the property $\widehat\psi(w_P)=1$ forces $w_P=1$, the subgroups
$\mathrm{Stab}_{W_P^{\mathrm{aff}}}\left(\la/M\right)$ and $\widehat\psi \left(\mathrm{Stab}_{W_P^{\mathrm{aff}}}\left(\la/M\right)\right) $ are isomorphic and hence
$$\sum_{w'\in \widehat\psi\left(\mathrm{Stab}_{W_P^{\mathrm{aff}}}\left(\frac{\la}{M}\right)\right)}\sigma(w')=\sum_{w_P\in \mathrm{Stab}_{W_P^{\mathrm{aff}}}\left(\frac{\la}{M}\right)}\sigma\circ \widehat\psi(w_P).$$
Definition \eqref{LP} of the weight set $ \Lambda^\sigma_{P,M}$ also ensures that for any $\la \in \Lambda^{\sigma}_{P,M}$ it holds that $\la/M\in F_P^\sigma $ and, taking into account defining relation \eqref{FsiP} and notation \eqref{hPM}, the final identity yielding relations \eqref{disortog} is derived,
$$\sum_{w_P\in \mathrm{Stab}_{W_P^{\mathrm{aff}}}\left(\frac{\la}{M}\right)}\sigma\circ \widehat\psi(w_P)=\sum_{w_P\in \mathrm{Stab}_{W_P^{\mathrm{aff}}}\left(\frac{\la}{M}\right)}1= h_{P,M}(\la). $$   
\end{proof}
\begin{thm}\label{Wartbas}
The functions $\varphi^\sigma_\la$, $\la\in\Lambda^\sigma_{P,M}$ form for any $M\in \N$, $M>m^\si$ an orthogonal basis of the Hilbert space $\mathcal{H}_{Q^\vee,M}^\si$.
\end{thm}
\begin{proof}
Theorem \ref{disortogthm} grants linear independence of the set of discretized functions $\varphi^\sigma_\la: F_{Q^\vee,M}^\sigma \map \C$, $\la\in\Lambda^\sigma_{P,M}$, 
and Theorem \ref{pocet} guarantees that this set of orthogonal functions has the cardinality
$$\abs{\Lambda^\sigma_{P,M}}=\abs{F_{Q^\vee,M}^\sigma}=\dim\mathcal{H}_{Q^\vee,M}^\si. $$
\end{proof}
As a consequence of the discrete orthogonality of Weyl orbit functions, the discrete orthogonality of Hartley orbit functions is derived in the following theorem.
\begin{thm}
For any $\la,\la'\in\Lambda^\sigma_{P,M}$, it holds that
\begin{equation*}
\langle \hart^\sigma_\la,\hart^\sigma_{\la'}\rangle_M^\sigma=\abs{W}M^nh_{P,M}(\la)\,\delta_{\la\la'},
\end{equation*}
where the coefficients $h_{P,M}$ are defined by \eqref{hPM}.
\end{thm}
\begin{proof}
The following trigonometric identity, 
$$(\cos{\alpha}+\sin{\alpha})(\cos{\alpha'}+\sin{\alpha'})=\mathrm{Re}\left(e^{\i\alpha}e^{-\i\alpha'}\right)+\mathrm{Im}\left(e^{\i\alpha}e^{\i\alpha'}\right),$$
valid for any $\al, \al' \in \R$, implies together with discrete orthogonality relations \eqref{disortog} that
\begin{equation*}
\langle\hart^\sigma_\lambda,\hart^\sigma_{\lambda'}\rangle_M^\sigma=\mathrm{Re}\,\langle\varphi^\sigma_\lambda,\varphi^\sigma_{\lambda'}\rangle^\sigma_M+\mathrm{Im}\,\langle\varphi^\sigma_\lambda,\overline{\varphi^\sigma_{\lambda'}}\rangle_M^\sigma =\abs{W}M^nh_{P,M}(\la)\,\delta_{\la\la'} +\mathrm{Im}\,\langle\varphi^\sigma_\lambda,\overline{\varphi^\sigma_{\lambda'}}\rangle_M^\sigma.
\end{equation*}
Definition \eqref{Weylorb} immediately provides the following relation for complex conjugated function $\varphi^\sigma_{\lambda'},$
\begin{equation}\label{im}
\overline{\varphi^\sigma_{\lambda'}}=\varphi^\sigma_{-\lambda'},	
\end{equation}
and lattice property \eqref{P} of  the weight lattice $P$ ensures that $-\lambda'\in P$. Since the lattice $\frac{1}{M}P$ is $W_P^{\mathrm{aff}}-$invariant, there exist $\mu \in P\cap MF_P$ and $w_P \in W_P^{\mathrm{aff}}$ such that 
\begin{equation}\label{muM}
-\la'= M 	w_P \left(\frac{\mu}{M}\right).
\end{equation}
Relations \eqref{im}, \eqref{muM} and label symmetry \eqref{invphi} allow to calculate
\begin{equation*}
\mathrm{Im}\,\langle\varphi^\sigma_\lambda,\overline{\varphi^\sigma_{\lambda'}}\rangle_M^\sigma	=\mathrm{Im}\,\langle\varphi^\sigma_\lambda,\varphi^\sigma_{-\lambda'}\rangle_M^\sigma	
= \si\circ\widehat{\psi}(w_P)\cdot \mathrm{Im}\,\langle\varphi^\sigma_\lambda,\varphi^\sigma_{\mu}\rangle_M^\sigma.	
\end{equation*}
If $\mu\in \Lambda^\sigma_{P,M}$, then discrete orthogonality relations \eqref{disortog} ensure that $\mathrm{Im}\,\langle\varphi^\sigma_\lambda,\varphi^\sigma_{\mu}\rangle_M^\sigma=0$. If, on the other hand, $\mu \in M(F_P\setminus F_P^\si),$ then vanishing property \eqref{vanphi} grants directly that  $\langle\varphi^\sigma_\lambda,\varphi^\sigma_{\mu}\rangle_M^\sigma=0$.	
\end{proof}
The discrete orthogonality of Hartley orbit functions also generates Hartley version of Theorem \ref{Wartbas}.
\begin{thm}\label{Hartbas}
The functions $\hart^\sigma_\la$, $\la\in\Lambda^\sigma_{P,M}$ form for any $M\in \N$, $M>m^\si$ an orthogonal basis of the Hilbert space $\mathcal{H}_{Q^\vee,M}^\si$.
\end{thm}

\subsection{Discrete transforms}\

The interpolating function $I[f]^{\sigma}_M:\R^n\map \C$ of any function $f\in\mathcal{H}_{Q^\vee,M}^\si $ is defined as a linear combination of the basis functions $\phi^\sigma_\la,$
\begin{equation}
I[f]^{\sigma}_M(a)\equiv \sum_{\la\in \Lambda_{P,M}^\sigma} c^\sigma_\la \phi^\sigma_\la(a), \label{intc} 
\end{equation}
satisfying the condition
\begin{equation*}
I[f]^\sigma_M(s)= f(s), \q s\in F_{Q^\vee,M}^\sigma. \label{intcs}
\end{equation*}
The frequency spectrum coefficients $c^\sigma_\la$ in \eqref{intc} are  uniquely determined by Theorem \ref{Wartbas} and calculated as standard Fourier coefficients,
\begin{equation}\label{for}
c^\sigma_\la= \frac{\sca{f}{\phi^\sigma_\la}_M^\sigma}{\sca{\phi^\sigma_\la}{\phi^\sigma_\la}_M^\sigma}=(\abs{W} M^nh_{P,M}(\la))^{-1}\sum_{s\in F_{Q^\vee,M}^\sigma}\ep(s) f(s)\overline{\phi^\sigma_\la(s)},
\end{equation}
and the corresponding Plancherel formulas also hold
$$\sum_{s\in F_{Q^\vee,M}^\sigma}\ep(s)\abs{f(s)}^2=\abs{W}M^n\sum_{\la\in\Lambda_{P,M}^\sigma}h_{P,M}(\la)\abs{c_\lambda}^2.$$
Equations \eqref{for} and \eqref{intc} establish forward and backward discrete Fourier-Weyl transforms, respectively, of the function $f\in\mathcal{H}_{Q^\vee,M}^\si $.

Similarly, the Hartley interpolating function $H[f]^{\sigma}_M:\R^n\map \C$ of any function $f\in\mathcal{H}_{Q^\vee,M}^\si $ is defined as a linear combination of the Hartley basis functions $\hart^\sigma_\la,$
\begin{align}
H[f]^{\sigma}_M(a)\equiv& \sum_{\la\in \Lambda_{P,M}^\sigma} d^\sigma_\la \hart^\sigma_\la(a), \label{intd} 
\end{align}
satisfying the condition
\begin{align*}
H[f]^\sigma_M(s)=& f(s), \q s\in F_{Q^\vee,M}^\sigma. 
\end{align*}
The frequency spectrum coefficients $d^\sigma_\la$ in \eqref{intd} are  uniquely determined by Theorem \ref{Hartbas} and calculated as standard Fourier coefficients,
\begin{align}
d^\sigma_\la=& \frac{\sca{f}{\hart^\sigma_\la}_M^\sigma}{\sca{\hart^\sigma_\la}{\hart^\sigma_\la}_M^\sigma}=(\abs{W} M^nh_{P,M}(\la))^{-1}\sum_{s\in F_{Q^\vee,M}^\sigma}\ep(s) f(s)\hart^\sigma_\la(s),\label{dtrans}
\end{align}
and the corresponding Plancherel formulas also hold
$$\sum_{x\in F_{Q^\vee,M}^\sigma}\ep(s)\abs{f(s)}^2=\abs{W}M^n\sum_{\la\in\Lambda_{P,M}^\sigma}h_{P,M}(\la)\abs{d^\si_\lambda}^2.$$
Equations \eqref{dtrans} and \eqref{intd} establish forward and backward discrete Hartley-Weyl transforms, respectively, of the function $f\in\mathcal{H}_{Q^\vee,M}^\si $.

\section*{Concluding remarks}
\begin{itemize}\item The choice of the fundamental domain $F_P$ of the dual extended affine Weyl group is not unique. As demonstrated in Example \ref{exA3}, the Brillouin zone \cite{Mich} of $A_3$, intersected with the dominant Weyl chamber and with certain boundaries omitted, produces another viable fundamental domain for the $A_3$ case. Significant advantage of the presented Kac coordinates approach stems from existence of an effective algorithm for constructing the weight sets $\Lambda^\sigma_{P,M}$ from formula~\eqref{MFP} for any case. Indeed, explicit relations \eqref{FQS} and \eqref{fun1} immediately produce the weight sets $P\cap MF^\si_Q$. Sorting the weights from $P\cap MF^\si_Q$ into $\Gamma_M-$orbits, selecting the lexicographically highest in each orbit, while excluding those with negative sign homomorphism values of stabilizing $\gamma\in \Gamma_M$, yields directly the set  $\Lambda^\sigma_{P,M}$ for any fixed $M\in \N$.
\item 
The counting formulas for the case $A_n$ are due to their link to combinatorics of necklaces already present in various contexts in the mathematical literature. Indeed, the action of the cyclic group $\Gamma_M$ on the set $\Lambda_M$ implies that the number $|\Lambda^{\id}_{P,M}(A_n)|$ coincides with the number of necklaces with $n+1$ white and $M$ black beads. Moreover, conditions \eqref{condition2} for the point sets $F_{Q^\vee,M}^{\sigma}(A_n)$ are in fact special cases of equation (2) in \cite{Elash} and thus, preserving the notation from  \cite{Elash}, it holds that 
\begin{align}
\abs{F_{Q^\vee,M}^{\id }(A_{n})}&=a_{0}(n+1,M),\nonumber\\
\abs{F_{Q^\vee,M}^{\sigma^e }(A_{2k})}&=a_{0}(2k+1,M-2k-1),\label{Ram}\\
\abs{F_{Q^\vee,M}^{\sigma^e }(A_{2k+1})}&=a_{k+1}(2k+2,M-2k-2).\nonumber	
\end{align}
Using the H\"{o}lder's identity \cite{Hold} to evaluate the Ramanujan sums in the explicit expressions for $a_k(m,n)$ in \cite[Thm. 1]{Elash} allows to bring equations \eqref{Ram} to the form of $|\Lambda^{\sigma}_{P,M}(A_n)|$ counting formulas \eqref{pocetAn} and \eqref{SpocetAn}. Such direct comparison provides an alternative proof of the $A_n$ case of Theorem~\ref{pocet}. Note also that the counting formulas studied in conjunction with the perfect forms in  Lemma 5.1 in \cite{Barn} coincide with the counting formula for $|\Lambda^{\sigma^e}_{P,M}(A_3)|$.
\item Good behaviour of the dual-weight discretization of Weyl orbit functions in interpolation estimates \cite{HaHrPa2} indicates similar viability of the discrete transforms \eqref{for} and \eqref{dtrans} in various applications related to digital data processing. Interpolation performance of the dual-root lattice discretization and existence of interpolation convergence criteria \cite{xuAd} pose open problems. The four families of Weyl orbit functions induce four families of orthogonal polynomials \cite{MPcub,MMP} which are special cases of multivariate Jacobi and Macdonald polynomials \cite{HMjac}.  Existence of generalization of the dual-weight lattice orthogonality of the selected subset of Macdonald polynomials \cite{diejen} to the dual-root lattice poses an open problem. The related polynomial interpolation and approximation methods, cubature formulas \cite{HMPcub} and their comparison to the weight and dual weight versions deserve further study.
\item Besides developing novel discrete transforms on generalized and composed grids, other fundamentally different options are provided by existence of functions invariant with respect to some subgroups of the given Weyl group. These even normal subgroups of index 2 of the Weyl groups are generated as kernels of each of the sign homomorphisms. There exist six more types of special functions induced by the even subgroups, called $E-$functions, for root systems with two lengths of the roots \cite{HJ} and one type for root systems with one length of the roots \cite{MP2}. The dual-weight lattice discretization of all ten types of Weyl and Hartley orbit functions is derived in a unified manner and full generality in \cite{HJ}. Extending the present dual-root lattice Fourier calculus to all ten types of Weyl and Hartley orbit functions poses a deep unsolved problem.
\end{itemize}

\section*{Acknowledgements}

This work was supported by the Grant Agency of the Czech Technical University in Prague, grant number SGS16/239/OHK4/3T/14.  LM and JH gratefully acknowledge the support of this work by RVO14000.

\end{document}